\documentclass[12pt,draftclsnofoot,onecolumn]{IEEEtran}
\usepackage{enumerate}
\usepackage{eurosym}
\usepackage{amsmath}
\usepackage{amsfonts}
\usepackage{amssymb}
\usepackage{graphicx}
\usepackage{epstopdf}
\usepackage{setspace}
\usepackage{arydshln}
\usepackage{float}
\usepackage[ruled,linesnumbered]{algorithm2e}
\usepackage[noend]{algpseudocode}
\usepackage[english]{babel}
\usepackage{amsthm}
\usepackage{enumitem}
\usepackage{xcolor}
\usepackage{soul}
\usepackage{url}
\usepackage{booktabs}

\makeatletter
\renewcommand{\fnum@figure}{Fig. \thefigure}
\makeatother

\makeatletter
\newcommand{\removelatexerror}{\let\@latex@error\@gobble}
\makeatother

\theoremstyle{proposition}
\newtheorem{proposition}{Proposition}
\theoremstyle{remark}

\begin{document}

\title{Joint Design of Radar Waveform and Detector via End-to-end Learning with
Waveform Constraints}
\author{Wei Jiang, \IEEEmembership{Student Member, IEEE,} Alexander M.
	Haimovich, \IEEEmembership{Fellow, IEEE} and Osvaldo
	Simeone, \IEEEmembership{Fellow, IEEE}
\thanks{W. Jiang and A. M. Haimovich are with the Center for Wireless Information Processing (CWiP), Department of Electrical and Computer Engineering, New Jersey Institute of Technology, Newark, NJ 07102 USA (e-mail: wj34@njit.edu;
	haimovich@njit.edu).}
\thanks{Osvaldo Simeone is with the King's Communications, Information Processing $\&$ Learning (KCLIP) Lab, Department of Engineering, King’s College London, London WC2R 2LS, UK (e-mail: osvaldo.
	simeone@kcl.ac.uk). His work was supported by the European Research Council (ERC) under the European Union’s Horizon 2020 Research and
	Innovation Programme (Grant Agreement No. 725731).}}

\maketitle

\begin{abstract}
	
The problem of data-driven joint design of transmitted waveform and detector in a radar system is addressed in this paper. We propose two novel learning-based approaches to waveform and detector design based on end-to-end training of the radar system. The first approach consists of alternating supervised training
of the detector for a fixed waveform and reinforcement learning of the transmitter for a fixed detector. In the second approach, the transmitter and detector are trained simultaneously. Various operational waveform constraints, such as peak-to-average-power ratio (PAR) and spectral compatibility, are incorporated into the design. Unlike traditional radar design methods that rely on rigid mathematical models with limited applicability, it is shown that radar learning can be robustified by training the detector with synthetic data generated from multiple statistical models of the environment.
Theoretical considerations and results show that the proposed methods
are capable of adapting the transmitted waveform to environmental conditions while satisfying design constraints.
\end{abstract}


\begin{IEEEkeywords}
Waveform design, radar detector design, waveform constraints, reinforcement learning, supervised learning.
\end{IEEEkeywords}

\section{Introduction}
\subsection{Context and Motivation}
Design of radar waveforms and detectors has been a topic of great
interest to the radar community (see e.g. \cite{Kay1998}-\cite{Kay 2007}).
For best performance, radar waveforms and detectors should be designed jointly \cite{Richards 2010}, \cite{MU}. Traditional joint design of waveforms and detectors typically relies on
mathematical models of the environment, including targets, clutter,
and noise. In contrast, this paper proposes data-driven approaches based on end-to-end learning of radar systems, in which reliance on rigid mathematical models of targets, clutter and noise is relaxed.

Optimal detection in the Neyman-Pearon (NP) sense guarantees highest probability of detection for a specified probability of false alarm \cite{Kay1998}. 
The NP detection test relies on the likelihood (or log-likelihood) ratio, which is the ratio of probability density functions
(PDF) of the received signal conditioned on the presence or absence of a target.
Mathematical tractability of models of the radar environment
plays an important role in determining the ease of implementation of an optimal detector. For some target, clutter and noise models, the structure of optimal detectors is well known \cite{Van 2004}-\cite{Richards 2005}. For example, closed-form expressions of the NP test metric are available when the applicable models are Gaussian \cite{Richards 2005}, and, in some cases, even for
non-Gaussian models \cite{Sangston 1994}. 

However, in most cases involving non-Gaussian models, the structure of optimal detectors generally involves intractable numerical integrations, making the implementation of such detectors computationally intensive \cite{Gini 1997}, \cite{Sangston 1999}.
For instance, it is shown in \cite{Gini 1997} that the NP detector requires a numerical integration with respect to the texture variable of the K-distributed clutter, thus precluding a closed-form solution. Furthermore, detectors designed based on a specific mathematical model of environment suffer performance degradation when the actual environment differs from
the assumed model \cite{Farina 1986}, \cite{Farina 1992}. Attempts to robustify performance by designing optimal detectors based on mixtures of random variables quickly run aground due to mathematical intractability.

Alongside optimal detectors, optimal radar waveforms may also be designed based on
the NP criterion. Solutions are known for some simple target, clutter and
noise models (see e.g. \cite{Delong1967}, \cite%
{Kay 2007}). However, in most cases, waveform design based on direct
application of the NP criterion is intractable, leading to various
suboptimal approaches. For example, mutual information, J-divergence and
Bhattacharyya distance have been studied as objective functions for waveform
design in multistatic settings \cite{Kay 2009}-\cite{Jeong 2016}.

In addition to target, clutter and noise models,
waveform design may have to account for various operational constraints. For
example, transmitter efficiency may be improved by constraining the
peak-to-average-power ratio (PAR) \cite{DeMaio2011}-\cite{Wu 2018}. 
A different constraint relates to the requirement of coexistence of radar and communication
systems in overlapping spectral regions. The National
Telecommunications and Information Administration (NTIA) and Federal
Communication Commission (FCC) have allowed sharing of some of the radar frequency
bands with commercial communication systems \cite{NTIA}. In order to protect the communication systems from radar interference, radar waveforms should be
designed subject to specified compatibility constraints. The design of radar
waveforms constrained to share the spectrum with communications systems has recently developed into an active area of research with a growing body of
literature \cite{Aubry2016}-\cite{Tang2019}.

Machine learning has been successfully applied to solve problems for which
mathematical models are unavailable or too complex to yield optimal
solutions, in domains such as computer vision \cite{ML 1.1}, \cite{ML 1.2}
and natural language processing \cite{ML 2.1}, \cite{ML 2.2}. Recently,
a machine learning approach has been proposed for implementing the physical layer of communication systems. Notably, in \cite{OShea 2017}, it is proposed to jointly design the transmitter and receiver of communication systems via end-to-end learning.
Reference \cite{PAR_OFDM} proposes an
end-to-end learning-based approach for jointly minimizing PAR and bit error rate in
orthogonal frequency division multiplexing systems. This approach requires
the availability of a known channel model. For the case of an unknown
channel model, reference \cite{Aoudia 2019} proposes an alternating training
approach, whereby the transmitter is trained via
reinforcement learning (RL) on the basis of noiseless feedback from the receiver,
while the receiver is trained by supervised learning. In \cite{SPSA}, the
authors apply simultaneous perturbation stochastic optimization for
approximating the gradient of a transmitter's loss function. A detailed review of the state of the art can be found in \cite{osvaldo2} (see also \cite{osvaldo3}-\cite{osvaldo5} for recent work).

In the radar field, learning machines trained in a supervised manner based on a suitable loss function have been shown to approximate the performance of the NP detector \cite{Moya 2009}, \cite{Moya 2013}. As a
representative example, in \cite{Moya 2013}, a neural network trained in a supervised manner using data that includes Gaussian interference, has been shown to approximate the performance of the NP detector. Note that design of the NP detector requires express knowledge of the Gaussian nature of the interference, while the neural network is trained with data that happens to be Gaussian, but the machine has no prior knowledge of the statistical nature of the data.  

\subsection{Main contributions}
In this work, we introduce two learning-based approaches for the joint
design of waveform and detector in a radar system. Inspired by \cite{Aoudia 2019}, end-to-end learning of a radar system is implemented by alternating supervised learning of the detector for a fixed waveform, and RL-based learning of the transmitter for a fixed detector. 
In the second approach, the learning of the detector and waveform are executed simultaneously, potentially speeding up training in terms of required radar transmissions to yield the training samples compared alternating training. 
In addition, we extend the problem formulation to include training of waveforms with PAR or spectral compatibility constraints. 

The main contributions of this paper are summarized as follows:

\begin{enumerate}
\item We formulate a radar system architecture based on the training of the detector and the transmitted waveform, both implemented as feedforward
multi-layer neural networks.

\item We develop two end-to-end learning algorithms for detection and waveform generation. In the first learning algorithm, detector and transmitted waveform are trained alternately: For a fixed waveform, the detector is trained using supervised learning so as to approximate the NP detector; and for a fixed detector, the transmitted waveform is trained via policy gradient-based RL. In the second algorithm, the detector and transmitter are trained simultaneously.  

\item We extend learning algorithms to incorporate waveform constraints,
specifically PAR and spectral compatibility constraints.

\item We provide theoretical results that relate alternating and simultaneous training by computing the gradients of the loss functions optimized by both methods. 

\item We also provide theoretical results that justify the use of RL-based transmitter training by comparing the gradient used by this procedure with the gradient of the ideal model-based likelihood function.
\end{enumerate}

This work extends previous results presented in the conference version \cite
{Wei 2019NN}. In particular, reference \cite{Wei 2019NN} proposes a learning algorithm, whereby supervised training of the radar detector is alternated with RL-based training of the unconstrained transmitted waveforms. 
As compared to the conference version \cite{Wei 2019NN}, this paper studies also
a simultaneous training; it develops methods for learning radar waveforms with various operational waveform constraints; and it provides a theoretical results regarding the relationship between alternating training and simultaneous training, as well as regarding the adoption of RL-based training of the transmitter.

The rest of this paper is organized as follows. A detailed system
description of the end-to-end radar system is presented in Section II.
Section III proposes two iterative algorithms of jointly training the
transmitter and receiver. Section IV provides theoretical properties of gradients
Numerical results are reported in Section V. Finally, conclusions are drawn
in Section VI.

Throughout the paper bold lowercase and uppercase letters
represent vector and matrix, respectively. The conjugate, the transpose, and the conjugate transpose operator are denoted by the symbols $(\cdot)^{*}$, $(\cdot)^{T}$, and $(\cdot)^{H}$, respectively.
The notations $\mathbb{C}^{K}$ and $\mathbb{R}^{K}$ represent sets of $K$-dimensional vectors of complex and
real numbers, respectively. The notation $|\cdot |$
indicates modulus, $||\cdot ||$ indicates the Euclidean norm, and 
$\mathbb{E}_{x\sim p_{x}}\{\cdot \}$ indicates the expectation of the argument with respect to the distribution of the random variable $x\sim p_{x}$, respectively. 
$\Re(\cdot )$ and $\Im (\cdot )$ stand for real-part and imaginary-part of the
complex-valued argument, respectively. The letter $j$ represents the
imaginary unit, i.e., $j=\sqrt{-1}$. The gradient of a function $f$: $%
\mathbb{R}^{n}\rightarrow \mathbb{R}^{m}$ with respect to $\mathbf{x}%
\in \mathbb{R}^{n}$ is $\nabla _{\mathbf{x}}f(\mathbf{x})\in \mathbb{R}%
^{n\times m}$.
\section{Problem Formulation}

Consider a pulse-compression radar system that uses the baseband transmit signal
\begin{equation}
	x(t)=\sum_{k=1}^{K}y_k \zeta\big( t- [k-1]T_c\big), \label{eq: time tx signal}
\end{equation}
where $\zeta(t)$ is a fixed basic chip pulse, $T_c$ is the chip duration, and $\{y_k\}_{k=1}^K$ are complex deterministic coefficients. The vector $\mathbf{y}\triangleq[ y_1,\dots, y_K ]^T$ is referred to as the fast-time \emph{waveform} of the radar system, and is subject to design.

The backscattered baseband signal from a stationary point-like target is given by
\begin{equation}
	z(t)=\alpha x(t-\tau) + c(t) + n(t)  \label{eq: time rx signal},
\end{equation}
where $\alpha$ is the target complex-valued gain, accounting for target backscattering and channel propagation effects; $\tau$ represents the target delay, which is assumed to satisfy the target detectability condition condition $\tau >\!\!>KT_c$; $c(t)$ is the clutter component; and $n(t)$ denotes
signal-independent noise comprising an aggregate of thermal noise, interference, and jamming. The clutter component $c(t)$ associated with a detection test performed at $\tau=0$ may be expressed 
\begin{equation}
	c(t)=\sum_{g=-K+1}^{K-1}\gamma_g x\big( t- g T_c \big), \label{eq: time clutter}
\end{equation}
where $\gamma_g$ is the complex clutter scattering coefficient at time delay $\tau=0$ associated with the $g$th range cell relative to the cell under test. Following chip matched filtering with $\zeta^*(-t)$, and sampling at $T_c-$spaced time instants
$t=\tau + [k-1] T_c$ for $k\in \{1, \dots K\}$, the $K\times 1$ discrete-time received signal $\mathbf{z}=[z(\tau), z(\tau+T_c),
\dots, z(\tau + [K-1]T_c)]^T$ for the range cell under test containing a point target with complex amplitude $\alpha$, clutter and noise can be written as
\begin{equation}
\mathbf{z}=\alpha \mathbf{y} + \mathbf{c} + \mathbf{n}, \label{eq: rx}
\end{equation}
where $\mathbf{c}$ and $\mathbf{n}$ denote, respectively, the clutter and noise vectors.

Detection of the presence of a target in the range cell under test is formulated as the following binary hypothesis testing problem:
\begin{equation}
	\left\{ \begin{aligned} &\mathcal{H}_0:{\mathbf{z}}={\mathbf{c}}+{\mathbf{n}} \\
		&\mathcal{H}_1:{\mathbf{z}}=\alpha \mathbf{y}+{\mathbf{c}}+{\mathbf{n}}.
	\end{aligned} \right.  \label{eq:binary hypo}
\end{equation}
In traditional radar design, the golden standard for detection is provided by the NP criterion of maximizing the probability of detection for a given probability of false alarm. Application of the NP criterion leads to the likelihood ratio test 
\begin{equation}
	\Lambda(\mathbf{z})=\frac{p(\mathbf{z}|\mathbf{y}, \mathcal{H}_1)}{p(\mathbf{z}|\mathbf{y}, \mathcal{H}_0)}\mathop{\gtrless}_{\mathcal{H}_0}^{\mathcal{H}_1} T_{\Lambda}, \label{eq: lrt}
\end{equation}
where $\Lambda(\mathbf{z})$ is the likelihood ratio, and $T_{\Lambda}$ is the detection threshold set based on the probability of false alarm constraint \cite{Richards 2010}. The NP criterion is also the golden standard for designing a radar waveform that adapts to the given environment, although, as discussed earlier, a direct application of this design principle is often intractable.

The design of optimal detectors and/or waveforms under the NP criterion requires on channel models of the radar environment, namely, knowledge of the conditional probabilities $p(\mathbf{z}|\mathbf{y}, \mathcal{H}_i)$ for $i=\{0,1\}$. The channel model $p(\mathbf{z}|\mathbf{y}, \mathcal{H}_i)$ is the likelihood of the observation $\mathbf{z}$ conditioned on the transmitted waveform $\mathbf{y}$ and hypothesis $\mathcal{H}_i$. In the following, we introduce an end-to-end radar system in which the detector and waveform are jointly learned in a data-driven fashion.

\subsection{End-to-end radar system}
The end-to-end radar system illustrated in Fig. 1 comprises a transmitter
and a receiver that seek to detect the presence of a target. Transmitter and receiver are implemented as two
separate parametric functions $f_{\boldsymbol{\theta}_T}(\cdot)$ and $f_{
	\boldsymbol{\theta}_R}(\cdot)$ with trainable parameter vectors $\boldsymbol{
	\theta}_T$ and $\boldsymbol{\theta}_R$, respectively.

\begin{figure}
		\vspace{-3ex} \hspace{25ex} \includegraphics[width=1.3
	\linewidth]{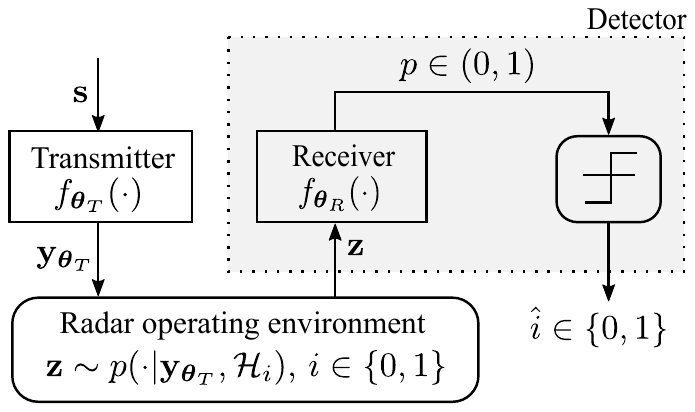} \vspace{-141ex}
	\caption{An end-to-end radar system operating over an unknown radar operating environment. Transmitter and receiver are implemented as two
		separate parametric functions $f_{\boldsymbol{\theta}_T}(\cdot)$ and $f_{
			\boldsymbol{\theta}_R}(\cdot)$ with trainable parameter vectors $\boldsymbol{
			\theta}_T$ and $\boldsymbol{\theta}_R$, respectively.}
	\label{f:end_to_end_real}
\end{figure}

As shown in Fig. \ref{f:end_to_end_real}, the input to the transmitter is a user-defined
initialization waveform ${\mathbf{s}}\in \mathbb{C}^{K}$. The transmitter outputs a radar waveform obtained through a trainable mapping $\mathbf{y}_{\boldsymbol{\theta}_T}=f_{\boldsymbol{\theta}_T}(\mathbf{s}) \in \mathbb{C}^K$.
The environment is modeled as a stochastic
system that produces the vector $\mathbf{z}\in \mathbb{C}^{K}$
from a conditional PDF $p(\mathbf{z}|\mathbf{y}_{\boldsymbol{\theta}_T}, \mathcal{H}_i)$
parameterized by a binary variable $i\in \{0,1\}$. The absence or presence of a target is indicated by the values $i=0$ and $i=1$ respectively, and hence $i$ is referred to as the \emph{target state indicator}. The receiver passes the received
vector $\mathbf{z}$ through a trainable mapping $p=f_{\boldsymbol{\theta}_R}(\mathbf{z})$, which produces the scalar $p\in (0,1)$. The final decision $\hat{i}\in \{0,1\}$ is made by comparing the output
of the receiver $p$ to a hard threshold in the interval $(0,1)$.

\subsection{Transmitter and Receiver Architectures}
As discussed in Section II-A, the transmitter and the receiver are implemented as two separate parametric functions $f_{\boldsymbol{\theta}_T}(\cdot)$ and $f_{\boldsymbol{\theta}_R}(\cdot)$. We now detail an implementation of the transmitter $f_{\boldsymbol{\theta}_T}(\cdot)$ and receiver $f_{\boldsymbol{\theta}_R}(\cdot )$ based on feedforward neural networks. 

A feedforward neural network is a parametric
function $\tilde{f}_{\boldsymbol{\theta}}(\cdot )$ that maps an input real-valued
vector $\mathbf{u}_{\text{in}}\in \mathbb{R}^{N_{\text{in}}}$ to an output real-valued
vector $\mathbf{u}_{\text{out}}\in \mathbb{R}^{N_{\text{out}}}$ via $L$ successive layers,
where $N_{\text{in}}$ and $N_{\text{out}}$ represent, respectively, the number of neurons of
the input and output layers. Noting that the input to the $l$th layer is the output of the $(l-1)$th layer, the output of the $l$th layer is given by
\begin{equation}
\mathbf{u}_{l}=\tilde{f}_{\boldsymbol{\theta}^{[l]}}(\mathbf{u}_{l-1})=\phi \big(%
\mathbf{W}^{[l]}\mathbf{u}_{l-1}+\mathbf{b}^{[l]}\big),\text{ }\text{for}%
\text{ }l=1,\dots ,L,
\end{equation}%
where $\phi (\cdot )$ is an element-wise activation function, and $\boldsymbol{\theta}^{[l]}=\{\mathbf{W}^{[l]},\mathbf{b}^{[l]}\}$ contains the trainable parameter of the $l$th layer comprising the weight $\mathbf{W}^{[l]}$ and bias $%
\mathbf{b}^{[l]}$. The vector of trainable parameters of the entire neural network comprises the parameters of all layers, i.e., $\boldsymbol{\theta }=\text{vec}\{\boldsymbol{\theta}^{[1]},\cdots,\boldsymbol{\theta}^{[L]}\}$.

The architecture of the end-to-end radar system with transmitter and receiver implemented based on feedforward neural networks is shown in Fig. \ref{f: arch}. The transmitter applies a complex initialization waveform $\mathbf{s}$ to
the function $f_{\boldsymbol{\theta }_{T}}(\cdot)$. The complex-value input $\mathbf{s}$ is processed by a complex-to-real conversion layer. This is followed by
a real-valued neural network $\tilde{f}_{\boldsymbol{\theta}_T}(\cdot)$. The output of the neural network is converted back to complex-values, and an output layer normalizes the transmitted power. As a result, the transmitter generates the radar waveform $\mathbf{y}_{\boldsymbol{\theta}_T}$.

The receiver applies the received signal $\mathbf{z}$ to the
function $f_{\boldsymbol{\theta }_{R}}(\cdot )$. Similar to the transmitter, a first layer converts complex-valued to real-valued vectors. The neural network at the receiver is denoted $\tilde{f}_{\boldsymbol{\theta}_R}(\cdot)$. The task of the
receiver is to generate a scalar $p\in (0,1)$ that approximates the
posterior probability of the presence of a target conditioned on the
received vector $\mathbf{z}$. To this end, the last layer of the neural network $\tilde{f}_{\boldsymbol{\theta}_R}(\cdot)$ is selected as a logistic regression layer consisting of operating over a linear combination of outputs from the previous layer. The presence or absence of the target is determined based on the output of the receiver and a threshold set according to a false alarm constraint.   

\begin{figure}
	\vspace{-5ex} \hspace{17ex} \includegraphics[width=1.2\linewidth]{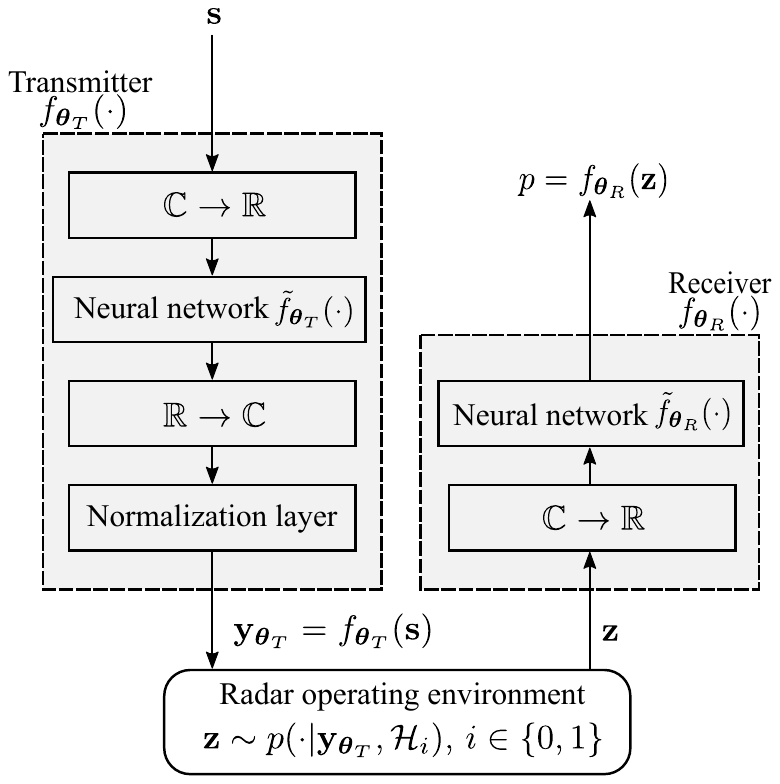}
	\vspace{-109ex}
	\caption{Transmitter and receiver architectures based on feedforward neural networks.}
	\label{f: arch}
\end{figure}

\section{Training of End-to-End Radar Systems}
This section discusses the joint optimization of the trainable parameter vectors $\boldsymbol{\theta }_{T}$ and $\boldsymbol{\theta }_{R}$ to meet application-specific performance requirements. Two training algorithms are proposed to train the end-to-end radar system. The first algorithm alternates between training of the receiver and of the transmitter. This algorithm is referred to as \emph{alternating training}, and is inspired by the approach used in \cite{Aoudia 2019} to train encoder and decoder of a digital communication system.
In contrast, the second algorithm trains the receiver and transmitter simultaneously. This approach is referred to as \emph{simultaneous training}. Note that the proposed two training algorithms are applicable to other differentiable parametric functions implementing the transmitter $f_{\boldsymbol{\theta }_{T}}(\cdot )$ and the receiver $f_{\boldsymbol{\theta }_{R}}(\cdot )$, such as recurrent neural network or its variants \cite{deeplearning}. In the following, we first discuss alternating training and then we detail simultaneous training.

\subsection{Alternating Training: Receiver Design}
Alternating training consists of iterations encompassing separate receiver and transmitter updates. In this subsection, we focus on the receiver updates. A receiver training update optimizes the receiver parameter vector $\boldsymbol{\theta }_{R}$ for a fixed transmitter waveform $\mathbf{y}_{\boldsymbol{\theta}_T}$. Receiver design is supervised in the sense that we assume the target state indicator $i$ to be available to the receiver during training.
Supervised training of the receiver for a fixed transmitter's parameter vector $\boldsymbol{\theta}_T$ is illustrated in Fig. \ref{f:rx_training}.
\begin{figure}[H]
	\vspace{-4ex} \hspace{16ex} \includegraphics[width=1.3\linewidth]{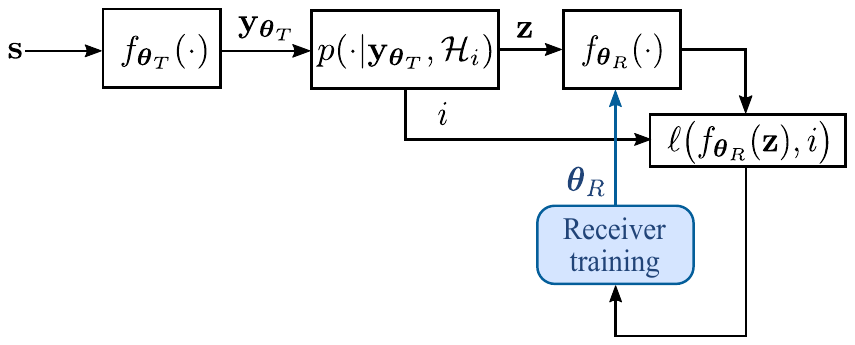}
	\vspace{-146ex}
	\caption{Supervised training of the receiver for a fixed transmitted waveform.}
	\label{f:rx_training}
\end{figure}

The standard cross-entropy loss \cite{Moya 2013} is adopted as the loss function for the receiver. For a given transmitted waveform $\mathbf{y}_{\boldsymbol{\theta}_T}=f_{\boldsymbol{\theta}_T}(\mathbf{s})$, the receiver average loss function is accordingly given by
\begin{equation}
	\begin{aligned}
		\mathcal{L}_R(\boldsymbol{\theta}_R)=&\sum_{i\in\{0,1\}}P(\mathcal{H}_i)\mathbb{E}_{\substack{ \mathbf{z}\sim p(\mathbf{z}|\mathbf{y}_{\boldsymbol{\theta}_T},\mathcal{H}_i)}}\big\{\ell \big( f_{\boldsymbol{\theta}_R}(\mathbf{z}),i\big)\big\}, 
	\end{aligned}\label{eq: rx loss}
\end{equation}
where $P(\mathcal{H}_i)$ is the prior probability of the target state indicator $i$, and $\ell\big( f_{\boldsymbol{\theta}_R}(\mathbf{z}),i\big)$ is the instantaneous cross-entropy loss for a pair $\big(f_{\boldsymbol{\theta}_R}(\mathbf{z}), i\big)$, namely,
\begin{equation}
	\ell\big( f_{\boldsymbol{\theta}_R}(\mathbf{z}),i\big)=-i\ln f_{\boldsymbol{\theta}_{R}}(\mathbf{z})-(1-i)\ln\big[1- f_{\boldsymbol{\theta}_{R}}(\mathbf{z})\big]. \label{eq: loss inst}
\end{equation}

For a fixed transmitted waveform, the receiver parameter vector $\boldsymbol{\theta}_R$ should be ideally optimized by minimizing  (\ref{eq: rx loss}), e.g., via gradient descent or one of its variants \cite{SGD}. 
The gradient of average loss (\ref{eq: rx loss}) with respect to the receiver parameter vector $\boldsymbol{\theta}_R$ is 
\begin{equation}
	{\nabla}_{\boldsymbol{\theta}_R}\mathcal{L}_R(\boldsymbol{\theta}_R)=\sum_{i\in\{0,1\}}P(\mathcal{H}_i)\mathbb{E}_{\substack{ \mathbf{z}\sim p(\mathbf{z}|\mathbf{y}_{\boldsymbol{\theta}_T},\mathcal{H}_i)}}\big\{ {\nabla}_{\boldsymbol{\theta}_R} \ell\big( f_{\boldsymbol{\theta}_R}(\mathbf{z}),i\big)\big\}.  \label{eq: rx loss grad.}
\end{equation} 
This being a data-driven approach, rather than assuming known prior probability of the target state indicator $P(\mathcal{H}_i)$ and likelihood $p(\mathbf{z}|\mathbf{y}_{\boldsymbol{\theta}_T},\mathcal{H}_i)$, the receiver is assumed to have access to $Q_R$ independent and identically distributed (i.i.d.) samples $\mathcal{D}_R=\big\{ \mathbf{z}^{(q)}\sim p(\mathbf{z}|\mathbf{y}_{\boldsymbol{\theta}_T},\mathcal{H}_{i^{(q)}}), {i^{(q)}}\in\{0,1\}  \big\}_{q=1}^{Q_R}$. 

Given the output of the receiver function $f_{\boldsymbol{\theta}_R}(\mathbf{z}^{(q)})$ for a received sample vector $\mathbf{z}^{(q)}$ and the indicator $i^{(q)}\in \{0,1 \}$, the instantaneous cross-entropy loss is computed from (\ref{eq: loss inst}), and the estimated receiver gradient is given by
\begin{equation}
	{\nabla}_{\boldsymbol{\theta}_R}\widehat{\mathcal{L}}_R(\boldsymbol{\theta}_R)=\frac{1}{Q_R}\sum_{q=1}^{Q_R} {\nabla}_{\boldsymbol{\theta}_R} \ell \big( f_{\boldsymbol{\theta}_R}(\mathbf{z}^{(q)}),{i^{(q)}} \big). \label{eq: est. rx loss grad}
\end{equation}
Using (\ref{eq: est. rx loss grad}), the receiver parameter vector $\boldsymbol{\theta}_R$ is adjusted according to stochastic gradient descent updates 
\begin{equation}
	\boldsymbol{\theta }_R^{(n+1)}=\boldsymbol{\theta}_R^{(n)}
	-\epsilon {\nabla}_{\boldsymbol{\theta}_R}\widehat{\mathcal{L}}_R(\boldsymbol{\theta}_R^{(n)})
	 \label{eq: rx sgd}
\end{equation}
across iterations $n=1,2,\cdots$, where $\epsilon >0$ is the learning rate.

\subsection{Alternating Training: Transmitter Design}
In the transmitter training phase of alternating training, the receiver parameter vector $\boldsymbol{\theta}_R$ is held constant, and the function $f_{\boldsymbol{\theta}_T}(\cdot)$ implementing the transmitter is optimized.
The goal of transmitter training is to find an optimized parameter vector $\boldsymbol{\theta}_T$ that minimizes the cross-entropy loss function (\ref{eq: rx loss}) seen as a function of $\boldsymbol{\theta}_T$. 

As illustrated in Fig. \ref{f:tx_training}, a stochastic transmitter outputs a waveform $\mathbf{a}$ drawn from a distribution $\pi(\mathbf{a}|\mathbf{y}_{\boldsymbol{\theta}_T})$ conditioned on $\mathbf{y}_{\boldsymbol{\theta}_T}=f_{\boldsymbol{\theta}_T}(\mathbf{s})$. The introduction of the randomization $\pi(\mathbf{a}|\mathbf{y}_{\boldsymbol{\theta}_T})$ of the designed waveform $\mathbf{y}_{\boldsymbol{\theta}_T}$ is useful to enable exploration of the design space in a manner akin to standard RL policies. To train the transmitter, we aim to minimize the average cross-entropy loss
\begin{equation}
	\begin{aligned}
		\mathcal{L}^{\pi}_T(\boldsymbol{\theta}_T)=&\sum_{i\in\{0,1\}}P(\mathcal{H}_i)\mathbb{E}_{\substack{ \mathbf{a}\sim \pi (\mathbf{a}|\mathbf{y}_{\boldsymbol{\theta}_T}) \\ \mathbf{z}\sim p(\mathbf{z}|\mathbf{a},\mathcal{H}_i)}}\big\{\ell\big( f_{\boldsymbol{\theta}_R}(\mathbf{z}),i \big)\big\}. \label{eq: tx loss RL no constraint}
	\end{aligned} 
\end{equation}
Note that this is consistent with (\ref{eq: rx loss}), with the caveat that an expectation is taken over policy $\pi(\mathbf{a}|\mathbf{y}_{\boldsymbol{\theta}_T})$. This is indicated by the superscript ``$\pi$''. 

\begin{figure}[H]
	\vspace{-4ex} \hspace{15ex} \includegraphics[width=1.3\linewidth]{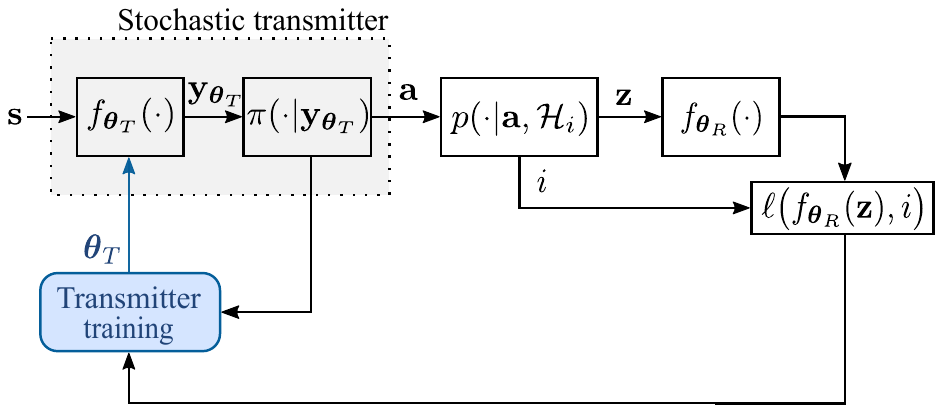}
	\vspace{-141ex}
	\caption{RL-based transmitter training for a fixed receiver design.}
	\label{f:tx_training}
\end{figure} 

Assume that the policy  $\pi(\mathbf{a}|\mathbf{y}_{\boldsymbol{\theta}_T})$ is differentiable with respect to the transmitter parameter vector $\boldsymbol{\theta}_T$, i.e., that the gradient $\nabla_{\boldsymbol{\theta }_T}\pi(\mathbf{a}|\mathbf{y}_{\boldsymbol{\theta}_T})$ exists.
The policy gradient theorem \cite{Sutton 2000} states that the gradient of the average loss (\ref{eq: tx loss RL no constraint}) can be written as
\begin{equation}
	\begin{aligned}
		\nabla_{\boldsymbol{\theta}_T}\mathcal{L}^{\pi}_T(\boldsymbol{\theta}_T)=&\sum_{i\in\{0,1\}}P(\mathcal{H}_i)\mathbb{E}_{\substack{ \mathbf{a}\sim \pi (\mathbf{a}|\mathbf{y}_{\boldsymbol{\theta}_T}) \\ \mathbf{z}\sim p(\mathbf{z}|\mathbf{a},\mathcal{H}_i)}}\big\{\ell \big( f_{\boldsymbol{\theta}_R}(\mathbf{z}), i\big)\nabla_{\boldsymbol{\theta}_T}\ln\pi(\mathbf{a}|\mathbf{y}_{\boldsymbol{\theta}_T})\big\}. \label{eq: tx loss RL grad} 
	\end{aligned}
\end{equation}
The gradient (\ref{eq: tx loss RL grad}) has the important advantage that it may be estimated via $Q_T$ i.i.d. samples  $\mathcal{D}_T=\big\{\mathbf{a}^{(q)}\sim \pi(\mathbf{a}|\mathbf{y}_{\boldsymbol{\theta}_T}), \mathbf{z}^{(q)}\sim p(\mathbf{z}|\mathbf{a}^{(q)},\mathcal{H}_{i^{(q)}}), i^{(q)}\in \{0,1\} \big\}_{q=1}^{Q_T}$, yielding the estimate
\begin{equation}
	{\nabla}_{\boldsymbol{\theta}_T}\widehat{\mathcal{L}}^{\pi}_T(\boldsymbol{\theta}_T)=\frac{1}{Q_T}\sum_{q=1}^{Q_T} \ell \big( f_{\boldsymbol{\theta}_R}(\mathbf{z}^{(q)}),i^{(q)}\big) \nabla_{\boldsymbol{\theta}_T}\ln\pi(\mathbf{a}^{(q)}|\mathbf{y}_{\boldsymbol{\theta}_T}). \label{eq: tx loss RL grad est} 
\end{equation}

With estimate (\ref{eq: tx loss RL grad est}), in a manner similar to (\ref{eq: rx sgd}), the transmitter parameter vector $\boldsymbol{\theta}_T$ may be optimized iteratively according to the stochastic gradient descent update rule
\begin{equation}
	\begin{aligned}
		&\boldsymbol{\theta }_T^{(n+1)}=\boldsymbol{\theta}_T^{(n)}
		-\epsilon {\nabla}_{\boldsymbol{\theta}_T}\widehat{\mathcal{L}}^{\pi}_T(\boldsymbol{\theta}_T^{(n)})
	\end{aligned} \label{eq: tx sgd}
\end{equation}
over iterations $n=1,2,\cdots$.
The alternating training algorithm is summarized as Algorithm 1.  The training process is carried out until a stopping criterion is satisfied. For example, a prescribed number of iterations may have been reached, or a number of iterations may have elapsed during which the training loss (\ref{eq: tx loss RL no constraint}), estimated using samples $\mathcal{D}_T$, may have not decreased by more than a given amount.

\DontPrintSemicolon%
\begin{algorithm}[]
	\SetAlgoLined
	\KwIn{initialization waveform $\mathbf{s}$; stochastic policy $\pi_{\boldsymbol{\theta}_T}(\cdot|\mathbf{y})$; learning rate $\epsilon$}
	\KwOut{learned parameter vectors $\boldsymbol{\theta}_R$ and $\boldsymbol{\theta}_T$}
	initialize $\boldsymbol{\theta}_R^{(0)}$ and $\boldsymbol{\theta}_T^{(0)}$, and set $n=0$\;
	\While{stopping criterion not satisfied}{
		\tcc{receiver training phase}
		evaluate the receiver loss gradient ${\nabla}_{\boldsymbol{\theta}_R}\widehat{\mathcal{L}}_R(\boldsymbol{\theta}_R^{(n)})$ from (\ref{eq: est. rx loss grad}) with $\boldsymbol{\theta}_T=\boldsymbol{\theta}_T^{(n)}$\;
		update receiver parameter vector $\boldsymbol{\theta}_R$ via 
		\begin{equation*}
			\boldsymbol{\theta }_R^{(n+1)}=\boldsymbol{\theta}_R^{(n)}
			-\epsilon {\nabla}_{\boldsymbol{\theta}_R}\widehat{\mathcal{L}}_R(\boldsymbol{\theta}_R^{(n)})
		\end{equation*}
		and stochastic transmitter policy turned off\;
		\tcc{transmitter training phase}
		evaluate the transmitter loss gradient ${\nabla}_{\boldsymbol{\theta}_T}\widehat{\mathcal{L}}^{\pi}_{T}(\boldsymbol{\theta}_T^{(n)})$ from (\ref{eq: tx loss RL grad est}) with $\boldsymbol{\theta}_R=\boldsymbol{\theta}_R^{(n+1)}$\;
		update transmitter parameter vector $\boldsymbol{\theta}_T$ via
		\begin{equation*}
			\boldsymbol{\theta }_T^{(n+1)}=\boldsymbol{\theta}_T^{(n)}
			-\epsilon {\nabla}_{\boldsymbol{\theta}_T}\widehat{\mathcal{L}}^{\pi}_{T}(\boldsymbol{\theta}_T^{(n)})
		\end{equation*}\;
		$n\leftarrow n+1$ }
	
	\caption{Alternating Training}
\end{algorithm}

\subsection{Transmitter Design with Constraints}
We extend the transmitter training discussed in the previous section to incorporate waveform constraints on PAR and spectral compatibility. To this end, we introduce penalty functions that are used to modify the training criterion (\ref{eq: tx loss RL no constraint}) to meet these constraints.
\subsubsection{PAR Constraint} Low PAR waveforms are preferred in radar systems due to hardware limitations related to waveform generation.
A lower PAR entails a lower dynamic range of the power amplifier, which in turn allows an increase in average transmitted power. The PAR of a radar waveform ${\mathbf{y}}_{\boldsymbol{\theta}_T}=f_{\boldsymbol{\theta}_T}(\mathbf{s})$ may be expressed
\begin{equation}
J_{\text{PAR}}(\boldsymbol{\theta}_T)=\frac{\underset{k=1,\cdots ,K}{\max }|%
	{y}_{_{\boldsymbol{\theta}_T}, k}|^{2}}{||{\mathbf{y}}_{\boldsymbol{\theta}_T}||^{2}/K},
\label{eq: PAPR complex}
\end{equation}%
which is bounded according to $1\leq J_{\text{PAR}}(\boldsymbol{\theta}_T)\leq K$.

\subsubsection{Spectral Compatibility Constraint}
A spectral constraint is imposed when a radar system is required to operate
over a spectrum partially shared with other systems such as wireless communication networks. Suppose there
are $D$ frequency bands $\{\Gamma _{d}\}_{d=1}^{D}$ shared by the radar and
by the coexisting systems, where $\Gamma _{d}=[f_{d,l},f_{d,u}]$, with $%
f_{d,l}$ and $f_{d,u}$ denoting the lower and upper normalized frequencies of the $d$th band, respectively. The amount of interfering energy generated by the radar waveform ${\mathbf{y}}_{\boldsymbol{\theta}_T}$ in the $d$th shared band is  
\begin{equation}
\int_{f_{d,l}}^{f_{d,u}}\left\vert \sum_{k=0}^{K-1}{y}%
_{_{\boldsymbol{\theta}_T}, k}e^{-j2\pi fk}\right\vert^{2}df={\mathbf{y}^{H}_{\boldsymbol{\theta}_T}}{\boldsymbol{%
		\Omega }}_{d}{\mathbf{y}_{\boldsymbol{\theta}_T}},  \label{eq: waveform energy}
\end{equation}%
where 
\begin{equation}
\begin{aligned} 
\big[ {\boldsymbol{\Omega}}_d \big]_{v,h}
& =\left\{ \begin{aligned}
&f_{d,u}-f_{d,l} \qquad\qquad\qquad\qquad\text{if }  v=h\\ &\frac{e^{j2\pi
		f_{d,u}(v-h)}-e^{j2\pi f_{d,l}(v-h)}}{j2\pi (v-h)} \quad \text{ if } v\neq h
\end{aligned} \right. \end{aligned}
\end{equation}
for $(v,h)\in \{1,\cdots,K\}^2$. Let ${\boldsymbol{\Omega }}=\sum_{d=1}^{D}\omega_{d}{\boldsymbol{\Omega 
}}_{d}$ be a weighted interference covariance matrix, where the weights $\{\omega _{d}\}_{d=1}^{D}$ are assigned
based on practical considerations regarding the impact of interference in the $D$ bands. These include distance between the radar
transmitter and interferenced systems, and tactical importance of the
coexisting systems \cite{Aubry2015}. Given a radar waveform $\mathbf{y}_{\boldsymbol{\theta}_T}=f_{\boldsymbol{\theta}_T}(\mathbf{s})$, we define the spectral compatibility penalty function as
\begin{equation}
	J_{\text{spectrum}}(\boldsymbol{\theta}_T)={\mathbf{y}}^{H}_{\boldsymbol{\theta}_T}{\boldsymbol{\Omega }}{\mathbf{y}_{\boldsymbol{\theta}_T}}, \label{eq: spectrum complex}
\end{equation}
which is the total interfering energy from the radar waveform produced on the shared frequency bands.

\subsubsection{Constrained Transmitter Design}
For a fixed receiver parameter vector $\boldsymbol{\theta}_R$, the average loss (\ref{eq: tx loss RL no constraint}) is modified by introducing a penalty function $J\in\{ J_{\text{PAR}}, J_{\text{spectrum}}\}$. Accordingly, we formulate the transmitter loss function, encompassing (\ref{eq: tx loss RL no constraint}), (\ref{eq: PAPR complex}) and (\ref{eq: spectrum complex}), as
\begin{equation}
	\begin{aligned}
		\mathcal{L}^{\pi}_{T,c}(\boldsymbol{\theta}_T)&=\mathcal{L}^{\pi}_T(\boldsymbol{\theta}_T)+\lambda J(\boldsymbol{\theta}_T)\\
		&=\sum_{i\in\{0,1\}}P(\mathcal{H}_i)\mathbb{E}_{\substack{ \mathbf{a}\sim \pi (\mathbf{a}|\mathbf{y}_{\boldsymbol{\theta}_T}) \\ \mathbf{z}\sim p(\mathbf{z}|\mathbf{a},\mathcal{H}_i)}}\big\{\ell \big( f_{\boldsymbol{\theta}_R}(\mathbf{z}),i \big)\big\}+\lambda J(\boldsymbol{\theta}_T). \label{eq: tx loss RL}
	\end{aligned} 
\end{equation}
where $\lambda$ controls the weight of the penalty $J(\boldsymbol{\theta}_T)$, and is referred to as the \emph{penalty parameter}. When the penalty parameter $\lambda$ is small, the transmitter is trained to improve its ability to adapt to the environment, while placing less emphasis on reducing the PAR level or interference energy from the radar waveform; and vice versa for large values of $\lambda$. Note that the waveform penalty function $J(\boldsymbol{\theta}_T)$ depends only on the transmitter trainable parameters $\boldsymbol{\theta}_T$. Thus, imposing the waveform constraint does not affect the receiver training.

It is straightforward to write the estimated version of the gradient (\ref{eq: tx loss RL}) with respect to $\boldsymbol{\theta}_T$ by introducing the penalty as
\begin{equation}
	\nabla_{\boldsymbol{\theta}_T}\widehat{\mathcal{L}}^{\pi}_{T,c}(\boldsymbol{\theta}_T)=\nabla_{\boldsymbol{\theta}_T}\widehat{\mathcal{L}}^{\pi}_T(\boldsymbol{\theta}_T)+\lambda \nabla_{\boldsymbol{\theta}_T}J(\boldsymbol{\theta}_T), \label{eq: est tx loss grad constraint}
\end{equation}
where the gradient of the penalty function $\nabla_{\boldsymbol{\theta}_T}J(\boldsymbol{\theta}_T)$ is provided in Appendix A.

Substituting (\ref{eq: tx loss RL grad est}) into (\ref{eq: est tx loss grad constraint}), we finally have the estimated gradient
\begin{equation}
	\nabla_{\boldsymbol{\theta}_T} \widehat{\mathcal{L}}^{\pi}_{T,c}(\boldsymbol{\theta}_T)=\frac{1}{Q_T}\sum_{q=1}^{Q_T} \ell \big( f_{\boldsymbol{\theta}_R}(\mathbf{z}^{(q)}),i^{(q)}\big) \nabla_{\boldsymbol{\theta}_T}\ln\pi(\mathbf{a}^{(q)}|\mathbf{y}_{\boldsymbol{\theta}_T})+\lambda \nabla_{\boldsymbol{\theta}_T} J(\boldsymbol{\theta}_T), \label{eq: est tx loss grad constraint2}
\end{equation}
which is used in the stochastic gradient update rule
\begin{equation}
	\begin{aligned}
		&\boldsymbol{\theta }_T^{(n+1)}=\boldsymbol{\theta}_T^{(n)}
		-\epsilon {\nabla}_{\boldsymbol{\theta}_T}\widehat{\mathcal{L}}^{\pi}_{T,c}(\boldsymbol{\theta}_T^{(n)}) 	\quad \text{for }n=1,2,\cdots.
	\end{aligned} \label{eq: tx constraint_sgd}
\end{equation}

\subsection{Simultaneous Training}
This subsection discusses simultaneous training, in which the receiver and transmitter are updated simultaneously as illustrated in Fig. \ref{f:joint_training}. To this end, the objective function is the average loss 
\begin{equation}
	\begin{aligned}
		\mathcal{L}^{\pi}(\boldsymbol{\theta}_R, \boldsymbol{\theta}_T)=&\sum_{i\in\{0,1\}}P(\mathcal{H}_i)\mathbb{E}_{\substack{ \mathbf{a}\sim \pi (\mathbf{a}|\mathbf{y}_{\boldsymbol{\theta}_T}) \\ \mathbf{z}\sim p(\mathbf{z}|\mathbf{a},\mathcal{H}_i)}}\big\{\ell \big( f_{\boldsymbol{\theta}_R}(\mathbf{z}),i \big)\big\}. \label{eq: joint loss} 
	\end{aligned}
\end{equation}
This function is minimized over both parameters $\boldsymbol{\theta}_R$ and $\boldsymbol{\theta}_T$ via stochastic gradient descent. 
\begin{figure}[H]
	\vspace{-3ex} \hspace{16ex} \includegraphics[width=1.2\linewidth]{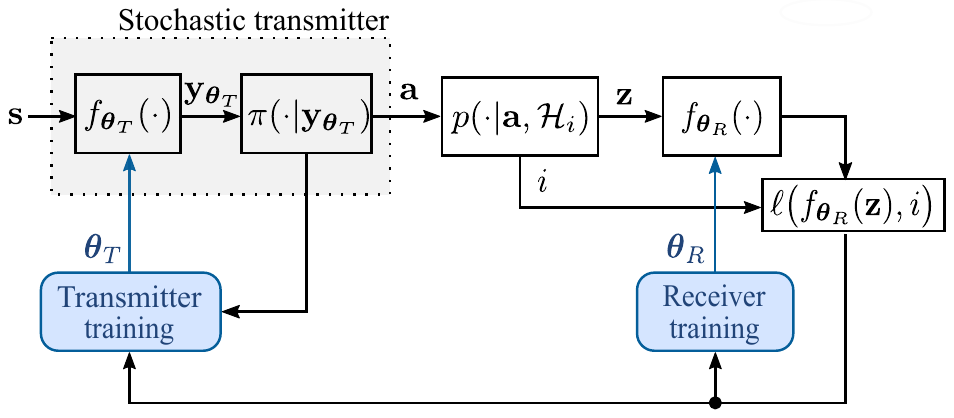}
	\vspace{-131ex}
	\caption{Simultaneous training of the end-to-end radar system. The receiver is trained by supervised learning, while the transmitter is trained by RL. 
	}
	\label{f:joint_training}
\end{figure}

The gradient of (\ref{eq: joint loss}) with respect to $\boldsymbol{\theta}_R$ is
\begin{equation}
	\nabla_{\boldsymbol{\theta}_R}\mathcal{L}^{\pi}(\boldsymbol{\theta}_R, \boldsymbol{\theta}_T)=\sum_{i\in\{0,1\}}P(\mathcal{H}_i)\mathbb{E}_{\substack{ \mathbf{a}\sim \pi (\mathbf{a}|\mathbf{y}_{\boldsymbol{\theta}_T}) \\ \mathbf{z}\sim p(\mathbf{z}|\mathbf{a},\mathcal{H}_i)}}\big\{\nabla_{\boldsymbol{\theta}_R}\ell\big( f_{\boldsymbol{\theta}_R}(\mathbf{z}),i \big)\big\}, \label{eq: rx loss grad joint} 
\end{equation}
and the gradient of (\ref{eq: joint loss}) with respect to $\boldsymbol{\theta}_T$ is
\begin{equation}
	\begin{aligned}
		\nabla_{\boldsymbol{\theta}_T}\mathcal{L}^{\pi}(\boldsymbol{\theta}_R, \boldsymbol{\theta}_T)=&\sum_{i\in\{0,1\}}P(\mathcal{H}_i)\nabla_{\boldsymbol{\theta}_T} \mathbb{E}_{\substack{ \mathbf{a}\sim \pi (\mathbf{a}|\mathbf{y}_{\boldsymbol{\theta}_T}) \\ \mathbf{z}\sim p(\mathbf{z}|\mathbf{a},\mathcal{H}_i)}}\big\{ \ell \big( f_{\boldsymbol{\theta}_R}(\mathbf{z}), i\big)\big\}\\
		=&\sum_{i\in\{0,1\}}P(\mathcal{H}_i)\mathbb{E}_{\substack{ \mathbf{a}\sim \pi (\mathbf{a}|\mathbf{y}_{\boldsymbol{\theta}_T}) \\ \mathbf{z}\sim p(\mathbf{z}|\mathbf{a},\mathcal{H}_i)}}\big\{ \ell \big( f_{\boldsymbol{\theta}_R}(\mathbf{z}), i\big)\nabla_{\boldsymbol{\theta}_T}\ln\pi(\mathbf{a}|\mathbf{y}_{\boldsymbol{\theta}_T})\big\}. \label{eq: tx loss RL grad joint} 
	\end{aligned}
\end{equation}

To estimate gradients (\ref{eq: rx loss grad joint}) and (\ref{eq: tx loss RL grad joint}), we assume access to $Q$ i.i.d. samples $\mathcal{D}=\big\{\mathbf{a}^{(q)}\sim \pi(\mathbf{a}|\mathbf{y}_{\boldsymbol{\theta}_T}), \mathbf{z}^{(q)}\sim p(\mathbf{z}|\mathbf{a}^{(q)},\mathcal{H}_{i^{(q)}}), i^{(q)}\in \{0,1\} \big\}_{q=1}^{Q}$. From (\ref{eq: rx loss grad joint}), the estimated receiver gradient is
\begin{equation}
	\nabla_{\boldsymbol{\theta}_R}\widehat{\mathcal{L}}^{\pi}(\boldsymbol{\theta}_R, \boldsymbol{\theta}_T)=\frac{1}{Q}\sum_{q=1}^{Q}\nabla_{\boldsymbol{\theta}_R}\ell\big( f_{\boldsymbol{\theta}_R}(\mathbf{z}^{(q)}),i ^{(q)}\big). \label{eq: rx loss grad joint est}
\end{equation}
Note that, in (\ref{eq: rx loss grad joint est}), the received vector $\mathbf{z}^{(q)}$ is obtained based on a given waveform $\mathbf{a}^{(q)}$ sampled from policy $\pi(\mathbf{a}|\mathbf{y}_{\boldsymbol{\theta}_T})$. Thus, the estimated receiver gradient (\ref{eq: rx loss grad joint est}) is averaged over the stochastic waveforms $\mathbf{a}$. This is in contrast to alternating training, in which the receiver gradient depends directly on the transmitted waveform $\mathbf{y}_{\boldsymbol{\theta}_T}$.

From (\ref{eq: tx loss RL grad joint}), the estimated transmitter gradient is given by
\begin{equation}
	\nabla_{\boldsymbol{\theta}_T}\widehat{\mathcal{L}}^{\pi}(\boldsymbol{\theta}_R, \boldsymbol{\theta}_T)=\frac{1}{Q}\sum_{q=1}^{Q} \ell \big( f_{\boldsymbol{\theta}_R}(\mathbf{z}^{(q)}),i^{(q)}\big) \nabla_{\boldsymbol{\theta}_T}\ln\pi(\mathbf{a}^{(q)}|\mathbf{y}_{\boldsymbol{\theta}_T}). \label{eq: tx loss RL grad joint est}
\end{equation}
Finally, denote the parameter set $\boldsymbol{\theta}=\{\boldsymbol{\theta }_{R}, \boldsymbol{\theta }_{T} \}$, from (\ref{eq: rx loss grad joint est}) and (\ref{eq: tx loss RL grad joint est}), the trainable parameter set $\boldsymbol{\theta}$ is updated according to the stochastic gradient descent rule 
\begin{equation}
	\boldsymbol{\theta}^{(n+1)}=\boldsymbol{\theta}^{(n)}
	-\epsilon {\nabla }_{\boldsymbol{\theta }}\widehat{\mathcal{L}}^{\pi}(\boldsymbol{\theta}_R^{(n)}, \boldsymbol{\theta}_T^{(n)})  \label{eq: sgd}
\end{equation} 
across iterations $n=1,2,\cdots.$

The simultaneous training algorithm is summarized in Algorithm 2. Like alternating training, simultaneous training can be directly extended to incorporate prescribed waveform constraints by adding the penalty term $\lambda J(\boldsymbol{\theta}_T)$ to the average loss (\ref{eq: joint loss}).
\DontPrintSemicolon%
\begin{algorithm}
	\SetAlgoLined
	\KwIn{initialization waveform $\mathbf{s}$; stochastic policy $\pi(\cdot|\mathbf{y}_{\boldsymbol{\theta}_T})$; learning rate $\epsilon$}
	\KwOut{learned parameter vectors $\boldsymbol{\theta}_R$ and $\boldsymbol{\theta}_T$}
	initialize $\boldsymbol{\theta}_R^{(0)}$ and $\boldsymbol{\theta}_T^{(0)}$, and set $n=0$\;
	\While{stopping criterion not satisfied}{
		evaluate the receiver gradient $\nabla_{\boldsymbol{\theta}_R}\widehat{\mathcal{L}}^{\pi}(\boldsymbol{\theta}_R^{(n)}, \boldsymbol{\theta}_T^{(n)})$ and the transmitter gradient $\nabla_{\boldsymbol{\theta}_T}\widehat{\mathcal{L}}^{\pi}(\boldsymbol{\theta}_R^{(n)}, \boldsymbol{\theta}_T^{(n)})$ from (\ref{eq: rx loss grad joint est}) and (\ref{eq: tx loss RL grad joint est}), respectively \;
		
		update receiver parameter vector $\boldsymbol{\theta}_R$ and transmitter parameter vector $\boldsymbol{\theta}_T$ simultaneously via
		\begin{equation*}
			\boldsymbol{\theta}_R^{(n+1)}=\boldsymbol{\theta}_R^{(n)}
			-\epsilon {\nabla }_{\boldsymbol{\theta }_R}\widehat{\mathcal{L}}^{\pi}(\boldsymbol{\theta}_R^{(n)}, \boldsymbol{\theta}_T^{(n)})
		\end{equation*}
	 	and
	 	\begin{equation*}
	 		\boldsymbol{\theta}_T^{(n+1)}=\boldsymbol{\theta}_T^{(n)}
	 		-\epsilon {\nabla }_{\boldsymbol{\theta }_T}\widehat{\mathcal{L}}^{\pi}(\boldsymbol{\theta}_R^{(n)}, \boldsymbol{\theta}_T^{(n)})
	 	\end{equation*}\;
		$n\leftarrow n+1$ }
	\caption{Simultaneous Training}
\end{algorithm}

\section{Theoretical properties of the gradients}

In this section, we discuss two useful theoretical properties of the gradients used for learning receiver and transmitter.
\subsection{Receiver Gradient}
As discussed previously, end-to-end learning of transmitted waveform and detector may be accomplished either by alternating or simultaneous training.
The main difference between alternating and simultaneous training concerns the update of the receiver trainable parameter vector $\boldsymbol{\theta}_R$. Alternating training of $\boldsymbol{\theta}_R$ relies on a fixed waveform  $\mathbf{y}_{\boldsymbol{\theta}_T}$ (see Fig. \ref{f:rx_training}), while simultaneous training relies on random waveforms $\mathbf{a}$ generated in accordance with a preset policy, i.e., $\mathbf{a} \sim \pi(\mathbf{a}|\mathbf{y}_{\boldsymbol{\theta}_T})$, as shown in Fig. \ref{f:joint_training}. The relation between the gradient applied by alternating training, $\nabla_{\boldsymbol{\theta }_R}{\mathcal{L}}_R(\boldsymbol{\theta}_R)$, and the gradient of simultaneous training, $\nabla_{\boldsymbol{\theta}_R}L^{\pi}(\boldsymbol{\theta}_R, \boldsymbol{\theta}_T)$,  with respect to $\boldsymbol{\theta}_R$ is stated by the following proposition.
\begin{proposition} 
	For the loss function (\ref{eq: rx loss}) computed based on a waveform $\mathbf{y}_{\boldsymbol{\theta}_T}$ and loss function (\ref{eq: tx loss RL no constraint}) computed based on a stochastic policy $\pi(\mathbf{a}|\mathbf{y}_{\boldsymbol{\theta}_T})$ continuous in $\mathbf{a}$, the following equality holds:
	\begin{equation}
		\begin{aligned}
			\nabla_{\boldsymbol{\theta }_R}{\mathcal{L}}_R(\boldsymbol{\theta}_R)=\nabla_{\boldsymbol{\theta}_R}\mathcal{L}^{\pi}(\boldsymbol{\theta}_R, \boldsymbol{\theta}_T).
		\end{aligned}
		\label{eq: Rx joint grad}
	\end{equation}		
\end{proposition}
\begin{proof}
	See Appendix B.
\end{proof}

Proposition 1 states that the gradient of simultaneous training, $\nabla_{\boldsymbol{\theta}_R}\mathcal{L}^{\pi}(\boldsymbol{\theta}_R, \boldsymbol{\theta}_T)$, equals the gradient of alternating training, $\nabla_{\boldsymbol{\theta }_R}\mathcal{L}_R(\boldsymbol{\theta}_R)$, even though simultaneous training applies a random waveform $\mathbf{a}\sim \pi(\mathbf{a}|\mathbf{y}_{\boldsymbol{\theta}_T})$ to train the receiver. Note that this result applies only to ensemble means according to (\ref{eq: rx loss}) and (\ref{eq: rx loss grad joint}), and not to the empirical estimates used by Algorithms 1 and 2. Nevertheless, Proposition 1 suggests that training updates of the receiver are unaffected by the choice of alternating or simultaneous training. That said, given the distinct updates of the transmitter's parameter, the overall trajectory of the parameters ($\boldsymbol{\theta}_R$, $\boldsymbol{\theta}_T$) during training may differ according to the two algorithms.
\subsection{Transmitter gradient}
As shown in the previous section, the gradients used for learning receiver parameters $\boldsymbol{\theta}_R$ by alternating training (\ref{eq: est. rx loss grad}) or simultaneous training (\ref{eq: rx loss grad joint est}) may be directly estimated from the channel output samples $\mathbf{z}^{(q)}$. In contrast, the gradient used for learning transmitter parameters $\boldsymbol{\theta}_T$ according to (\ref{eq: rx loss}) cannot be directly estimated from the channel output samples. To obviate this problem, in Algorithms 1 and 2, the transmitter is trained
by exploring the space of transmitted waveforms according to a policy $\pi(\mathbf{a}|\mathbf{y}_{\boldsymbol{\theta}_T})$. We refer to the transmitter loss gradient obtained via policy gradient (\ref{eq: tx loss RL grad joint}) as the \emph{RL transmitter gradient}. The benefit of RL-based transmitter training is that it renders unnecessary access to the likelihood function $p(\mathbf{z}|\mathbf{y}_{\boldsymbol{\theta}_T}, \mathcal{H}_i)$ to evaluate the RL transmitter gradient, rather the gradient is estimated via samples. We now formalize the relation of the RL transmitter gradient (\ref{eq: tx loss RL grad joint}) and the transmitter gradient for a known likelihood obtained according to (\ref{eq: rx loss}).

As mentioned, if the likelihood $p(\mathbf{z}|\mathbf{y}_{\boldsymbol{\theta}_T},\mathcal{H}_i)$ were known, and if it were differentiable with respect to the transmitter parameter vector $\boldsymbol{\theta}_T$, the transmitter parameter vector $\boldsymbol{\theta}_T$ may be learned by minimizing the average loss (\ref{eq: rx loss}), which we rewrite as a function of both $\boldsymbol{\theta}_R$ and $\boldsymbol{\theta}_T$ as
\begin{equation}
	\mathcal{L}(\boldsymbol{\theta}_R,\boldsymbol{\theta}_T)= \sum_{i\in\{0,1\}}P(\mathcal{H}_i)\mathbb{E}_{\substack{ \mathbf{z}\sim p(\mathbf{z}|\mathbf{y}_{\boldsymbol{\theta}_T},\mathcal{H}_i)}}\big\{ \ell \big( f_{\boldsymbol{\theta}_R}(\mathbf{z}),i\big)\big\}. \label{eq: known loss}
\end{equation}
The gradient of (\ref{eq: known loss}) with respect to $\boldsymbol{\theta}_T$ is expressed as
\begin{equation}
	\begin{aligned}
		\nabla_{\boldsymbol{\theta}_T}\mathcal{L}(\boldsymbol{\theta}_R,\boldsymbol{\theta}_T )
		&=\sum_{i\in\{0,1\}}P(\mathcal{H}_i)\mathbb{E}_{\substack{ \mathbf{z}\sim p(\mathbf{z}|\mathbf{y}_{\boldsymbol{\theta}_T},\mathcal{H}_i)}}\big\{\ell\big( f_{\boldsymbol{\theta}_R}(\mathbf{z}),i\big) \nabla_{\boldsymbol{\theta}_T}\ln p(\mathbf{z}|\mathbf{y}_{\boldsymbol{\theta}_T},\mathcal{H}_i) \big\},
	\end{aligned} \label{eq: tx loss known grad}
\end{equation}
where the equality leverages the following relation
\begin{equation}
	\nabla_{\boldsymbol{\theta}_T}p(\mathbf{z}|\mathbf{y}_{\boldsymbol{\theta}_T},\mathcal{H}_i)=p(\mathbf{z}|\mathbf{y}_{\boldsymbol{\theta}_T},\mathcal{H}_i)\nabla_{\boldsymbol{\theta}_T}\ln p(\mathbf{z}|\mathbf{y}_{\boldsymbol{\theta}_T},\mathcal{H}_i). \label{eq: log-trick}
\end{equation}

The relation between the RL transmitter gradient $\nabla_{\boldsymbol{\theta}_T}\mathcal{L}^{\pi}(\boldsymbol{\theta}_R,\boldsymbol{\theta}_T)$ in (\ref{eq: tx loss RL grad joint}) and the transmitter gradient $\nabla_{\boldsymbol{\theta}_T}\mathcal{L}(\boldsymbol{\theta}_R,\boldsymbol{\theta}_T)$ in (\ref{eq: tx loss known grad}) is elucidated by the following proposition.
\begin{proposition}
	If likelihood function $p(\mathbf{z}|\mathbf{y}_{\boldsymbol{\theta}_T},\mathcal{H}_i)$ is differentiable with respect to the transmitter parameter vector $\boldsymbol{\theta}_T$ for $i\in\{0,1\}$, the following equality holds
	\begin{equation}
		\nabla_{\boldsymbol{\theta}_T}\mathcal{L}^{\pi}(\boldsymbol{\theta}_R,\boldsymbol{\theta}_T)=\nabla_{\boldsymbol{\theta}_T}\mathcal{L}(\boldsymbol{\theta}_R,\boldsymbol{\theta}_T).
	\end{equation}
\end{proposition}
\begin{proof}
	See Appendix C.
\end{proof}
Proposition 2 establishes that the RL transmitter gradient $\nabla_{\boldsymbol{\theta}_T}\mathcal{L}^{\pi}(\boldsymbol{\theta}_R,\boldsymbol{\theta}_T)$ equals the transmitter gradient $\nabla_{\boldsymbol{\theta}_T}\mathcal{L}(\boldsymbol{\theta}_R,\boldsymbol{\theta}_T)$ for any given receiver parameters $\boldsymbol{\theta}_R$. Proposition 2 hence
provides a theoretical justification for replacing the gradient $\nabla_{\boldsymbol{\theta}_T}\mathcal{L}(\boldsymbol{\theta}_R,\boldsymbol{\theta}_T)$ with the RL gradient $\nabla_{\boldsymbol{\theta}_T}\mathcal{L}^{\pi}(\boldsymbol{\theta}_R,\boldsymbol{\theta}_T)$ to perform transmitter training as done in Algorithms 1 and 2.

\section{Numerical Results}

This section first introduces the simulation setup, and then it presents numerical examples of waveform design and detection performance that compare the proposed data-driven methodology with existing model-based approaches. While simulation results presented in this section rely on various models of target, clutter and interference, this work expressly distinguishes data-driven learning from model-based design. Learning schemes rely solely on data and not on model information. In contrast, model-based design implies a system structure that is based on a specific and known model. Furthermore, learning may rely on synthetic data containing diverse data that is generated according to a variety of models. In contrast, model-based design typically relies on a single model. For example, as we will see, a synthetic dataset for learning may contain multiple clutter sample sets, each generated according to a different clutter model. Conversely, a single clutter model is typically assumed for model-based design.

\subsection{Models, Policy, and Parameters} 

\subsubsection{Models of target, clutter, and noise}
The target is stationary, and has a Rayleigh envelope, i.e., $\alpha\sim \mathcal{CN}(0,\sigma_{\alpha}^2)$. The noise has a zero-mean Gaussian distribution with the correlation matrix $[\boldsymbol{\Omega }_n]_{v,h}=\sigma _{n}^{2}\rho^{|v-h|}$ for $(v,h)\in \{1,\cdots,K\}^2$, where $\sigma_n^2$ is the noise power and $\rho$ is the one-lag correlation coefficient. The clutter vector in (\ref{eq: rx}) is the superposition of returns from $2K-1$ consecutive range cells, reflecting all clutter illuminated by the $K$-length signal as it sweeps in range across the target. Accordingly, the clutter vector may be expressed as
\begin{equation}
	{\mathbf{c}}=\sum_{\substack{ g=-K+1 }}^{K-1}{\gamma }_{g}\mathbf{J%
	}_{g}{\mathbf{y}},  \label{eq: clutter}
\end{equation}
where $\mathbf{J}_{g}$ represents the shifting matrix at the $g$th range cell with elements 
\begin{equation}
	\big[\mathbf{J}_{g}\big]_{v,h}=\left\{ \begin{aligned} &1 \quad \text{if} \quad v-h=g\\
		&0\quad \text{if} \quad v-h\neq g \end{aligned}\quad (v,h)\in \{1,\cdots
	,K\}^{2}\right. .
\end{equation}
The magnitude $|\gamma_g|$ of the $g$th clutter scattering coefficient is generated according to a Weibull distribution \cite{Richards 2010}
\begin{equation}
	p(|\gamma_g|)=\frac{\beta}{\nu^{\beta}}|\gamma_g|^{\beta-1}\exp\bigg( - \frac{|\gamma_g|^{\beta}}{\nu^{\beta}} \bigg), \label{eq: Weibull pdf}
\end{equation}
where $\beta$ is the shape parameter and $\nu$ is the scale parameter of the distribution. Let $\sigma_{\gamma_g}^2$ represent the power of the clutter scattering coefficient $\gamma_g$. The relation between $\sigma_{\gamma_g}^2$ and the Weibull distribution parameters $\{\beta,\nu\}$ is \cite{Farina 1987}
\begin{equation}
	\sigma_{\gamma_g}^2=\text{E}\{|{\gamma}_g|^2\}=\frac{2\nu^2}{\beta}\Gamma\bigg(\frac{2}{\beta}\bigg),
\end{equation}
where $\Gamma(\cdot)$ is the Gamma function. The nominal range of the shape parameter is $0.25\leq\beta\leq2$ \cite{shape}. In the simulation, the complex-valued clutter scattering coefficient $\gamma_g$ is obtained by multiplying a real-valued Weibull random variable $|\gamma_g|$ with the factor $\exp(j\psi_g)$, where $\psi_g$ is the phase of $\gamma_g$ distributed uniformly in the interval $(0,2\pi)$. 
When the shape parameter $\beta=2$, the clutter scattering coefficient $\gamma_g$ follows the Gaussian distribution $\gamma_g \sim \mathcal{CN}(0,\sigma_{\gamma_g}^2)$. Based on the assumed mathematical models of the target, clutter and noise, it can be shown that the optimal detector in the NP sense is the square law detector \cite{Richards 2005}, and the adaptive waveform for target detection can be obtained by maximizing the signal-to-clutter-plus-noise ratio at the receiver output at the time of target detection (see Appendix A of \cite{Wei 2019NN} for details). 

\subsubsection{Transmitter and Receiver Models}
Waveform generation and detection is implemented using feedforward neural networks as explained in Section II-B. The transmitter $\tilde{f}_{\boldsymbol{\theta}_T}(\cdot)$
is a feedforward neural network with four layers, i.e., an input layer with $2K$ neurons, two hidden layers with $M=24$ neurons, and an output layer with $2K$ neurons. The activation function is exponential linear unit (ELU) \cite{ELU}. 
The receiver  $\tilde{f}_{\boldsymbol{\theta}_R}(\cdot)$ is implemented as a feedforward neural network with four layers, i.e., an input layer with
$2K$ neurons, two hidden layers with $M$ neurons, and an output layer with one neuron. The sigmoid function is chosen as the activation function. The layout of transmitter and receiver networks is summarized in Table I.
\begin{table}[H]
	\caption{Layout of the transmitter and receiver networks}
	\label{table:1}
	\centering
	\resizebox{0.6\columnwidth}{!}{
	\begin{tabular}{@{}cccccccc@{}}\toprule
		\multicolumn{1}{c}{} & \multicolumn{3}{c}{Transmitter $\tilde{f}_{\boldsymbol{\theta}_T}(\cdot)$} & \phantom{a} & \multicolumn{3}{c}{Receiver $\tilde{f}_{\boldsymbol{\theta}_R}(\cdot)$} \\
		 \cmidrule{2-4} \cmidrule{6-8} 
		Layer& 1 & 2-3 & 4 && 1 & 2-3 & 4 \\
		Dimension& $2K$ & $M$ & $2K$ && $2K$ & $M$ & $1$ \\ 
		Activation& - & ELU & Linear && - & Sigmoid & Sigmoid \\ 
		\bottomrule
	\end{tabular}
	}
\end{table}

\subsubsection{Gaussian policy} 
A Gaussian policy $\pi(\mathbf{a}|\mathbf{y}_{\boldsymbol{\theta}_T})$ is adopted for RL-based transmitter training. Accordingly, the output of the stochastic transmitter follows a complex Gaussian distribution $\mathbf{a}\sim\pi(\mathbf{a}|\mathbf{y}_{\boldsymbol{\theta}_T})=\mathcal{CN}\big(\sqrt{1-\sigma^2_p}\mathbf{y}_{\boldsymbol{\theta}_T},\frac{\sigma^2_p}{K}\mathbf{I}_K\big)$, where the per-chip variance $\sigma^2_p$ is referred to as the \emph{policy hyperparameter}. When $\sigma^2_p=0$, the stochastic policy becomes deterministic \cite{Silver 2014}, i.e., the policy is governed by a Dirac function at $\mathbf{y}_{\boldsymbol{\theta}_T}$. In this case, the policy does not explore the space of transmitted waveforms, but it ``exploits'' the current waveform. At the opposite end, when $\sigma^2_p=1$, the output of the stochastic transmitter is independent of $\mathbf{y}_{\boldsymbol{\theta}_T}$, and the policy becomes 
zero-mean complex-Gaussian noise with covariance matrix $\mathbf{I}_K/K$. Thus, the policy hyperparameter $\sigma^2_p$ is selected in the range $(0,1)$, and its value sets a trade-off between exploration of new waveforms versus exploitation of current waveform. 

\subsubsection{Training Parameters}
The initialization waveform $\mathbf{s}$ is a linear frequency modulated pulse with $K=8$ complex-valued chips and chirp rate $R=(100\times10^3)/(40\times 10^{-6})$ Hz/s. Specifically, the $k$th chip of $\mathbf{s}$ is given by
\begin{equation}
	\mathbf{s}(k)=\frac{1}{\sqrt{K}}\exp \big\{ j\pi R \big( k/f_s\big)^2 \big\} 
\end{equation}
for $\forall k\in\{0,\dots,K-1\}$, where $f_s=200$ kHz. The signal-to-noise ratio (SNR) is defined as
\begin{equation}
\text{SNR}=10\log_{10}\bigg\{\frac{\sigma_{\alpha}^2}{\sigma_n^2}\bigg\}. \label{eq: SNR}
\end{equation}
Training was performed at $\text{SNR}=12.5$ dB. The clutter environment is uniform
with $\sigma_{\gamma_g}^2=-11.7$ dB, $\forall g\in\{-K+1,\dots, K-1\}$, such that the overall clutter power is $\sum_{g=-(K-1)}^{K-1}\sigma_{\gamma_g}^2=0$ dB. The noise power is $\sigma_n^2=0$ dB, and the one-lag correlation coefficient $\rho=0.7$. 
Denote $\beta_{\text{train}}$ and $\beta_{\text{test}}$ the shape parameters of the clutter distribution (\ref{eq: Weibull pdf}) applied in training and test stage, respectively. Unless stated otherwise, we set $\beta_{\text{train}}=\beta_{\text{test}}=2$. 

To obtain a balanced classification dataset, the training set is populated by samples belonging to either hypothesis with equal prior probability, i.e., $
P(\mathcal{H}_0)=P(\mathcal{H}_1)=0.5$. The number of training samples is set as $Q_R=Q_T=Q=2^{13}$ in the estimated gradients (\ref{eq: est. rx loss grad}), (\ref{eq: tx loss RL grad est}), (\ref{eq: rx loss grad joint est}), and (\ref{eq: tx loss RL grad joint est}). Unless stated otherwise, the policy parameter is set to $\sigma^2_p=10^{-1.5}$, and the penalty parameter is $\lambda=0$, i.e., there are no waveform constraints.
The Adam optimizer \cite{adam} is adopted to train the system over a number of iterations chosen by trial and error. The learning rate is  $\epsilon=0.005$.
In the testing phase, $2\times10^5$ samples are used to estimate the probability of false alarm ($P_{fa}$) under hypothesis $\mathcal{H}_0$, while $5\times10^4$ samples are used to estimate the probability of detection ($P_d$) under hypothesis $\mathcal{H}_1$. Receiver operating characteristic (ROC) curves are obtained via Monte Carlo simulations by varying the threshold applied at the output of the receiver. Results are obtained by averaging over fifty trials. Numerical results presented in this section assume simultaneous training, unless stated otherwise.

\subsection{Results and Discussion}

\subsubsection{Simultaneous Training vs Training with Known Likelihood}
We first analyze the impact of the choice of the policy hyperparameter $\sigma_p^2$ on the performance on the training set. Fig. \ref{f: var_loss} shows the empirical cross-entropy loss of simultaneous training versus the policy hyperparameter $\sigma^2_p$ upon the completion of the training process.
The empirical loss of the system training with a known channel (\ref{eq: known loss}) is plotted as a comparison. It is seen that there is an optimal policy parameter $\sigma^2_p$ for which the empirical loss of simultaneous training approaches the loss with known channel. As the policy hyperparameter $\sigma^2_p$ tends to $0$, the output of the stochastic transmitter $\mathbf{a}$ is close to the waveform $\mathbf{y}_{\boldsymbol{\theta}_T}$, which leads to no exploration of the space of transmitted waveforms. In contrast, when the policy parameter $\sigma^2_p$ tends to $1$, the output of the stochastic transmitter becomes a complex Gaussian noise with zero mean and covariance matrix $\mathbf{I}_K/K$. In both cases, the RL transmitter gradient is difficult to estimate accurately. 
\begin{figure}
	\centering
	\vspace{-3ex} 
	\includegraphics[width=0.7\linewidth]{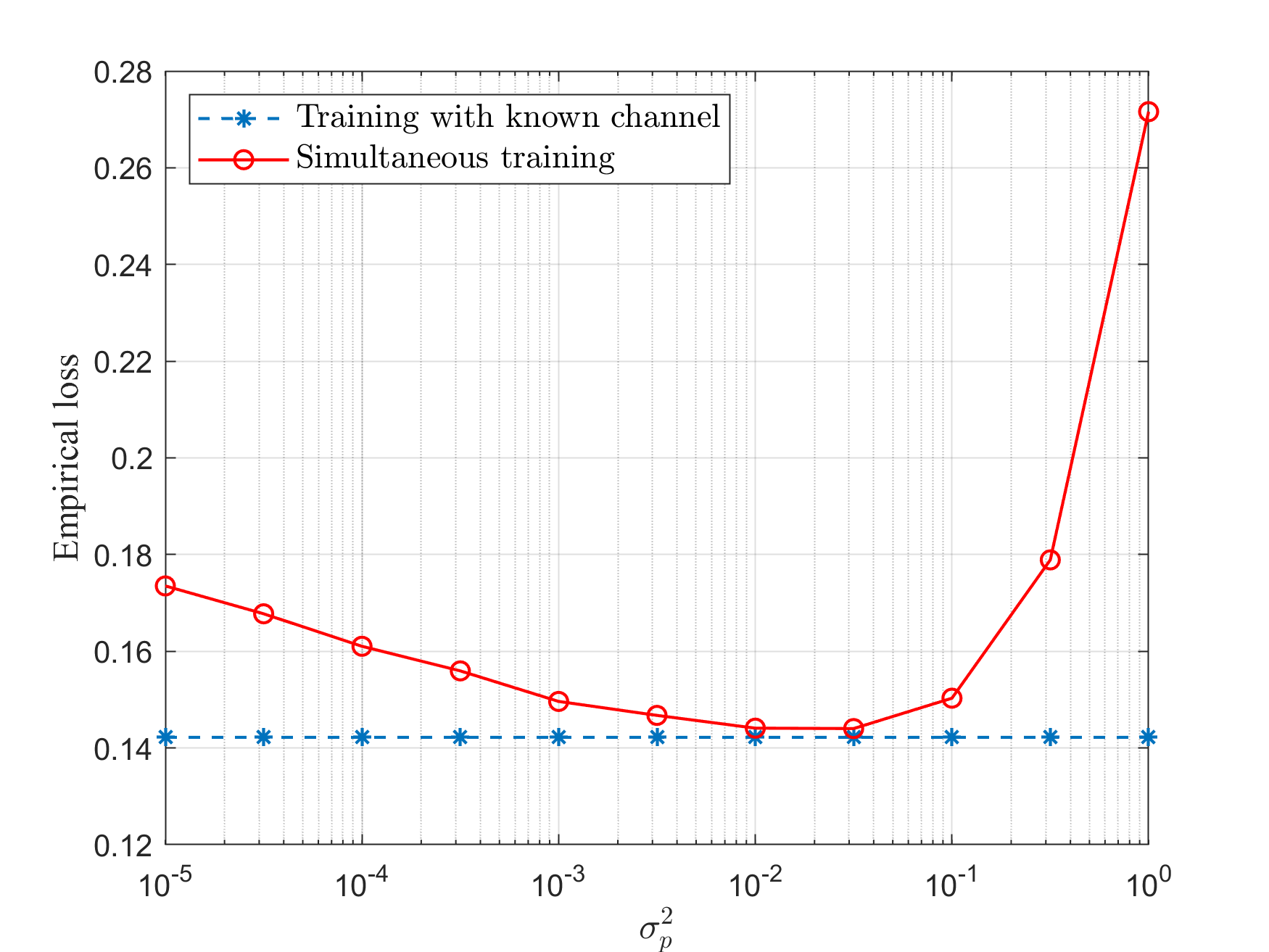} 
	\vspace{-3ex}  
	\caption{Empirical training loss versus  policy hyperparameter $\sigma^2_p$ for simultaneous training algorithm and training with known channel, respectively.} \label{f: var_loss}
\end{figure}

While Fig. \ref{f: var_loss} evaluates the performance on the training set in terms of empirical cross-entropy loss, the choice of the policy hyperparameter $\sigma^2_p$ should be based on validation data and in terms of the testing criterion that is ultimately of interest. To elaborate on this point,
ROC curves obtained by simultaneous training with different values of the policy hyperparameter $\sigma^2_p$ and training with known channel are shown in Fig. \ref{f: var_ROC}. As shown in the figure, simultaneous training with $\sigma^2_p=10^{-1.5}$ achieves a similar ROC as training with known channel. The choice $\sigma^2_p=10^{-1.5}$, also has the lowest empirical training loss in Fig. \ref{f: var_loss}. These results suggest that training is not subject to overfitting \cite{osvaldo1}.
\begin{figure}
	\centering
	\vspace{-3ex} 
	\includegraphics[width=0.7\linewidth]{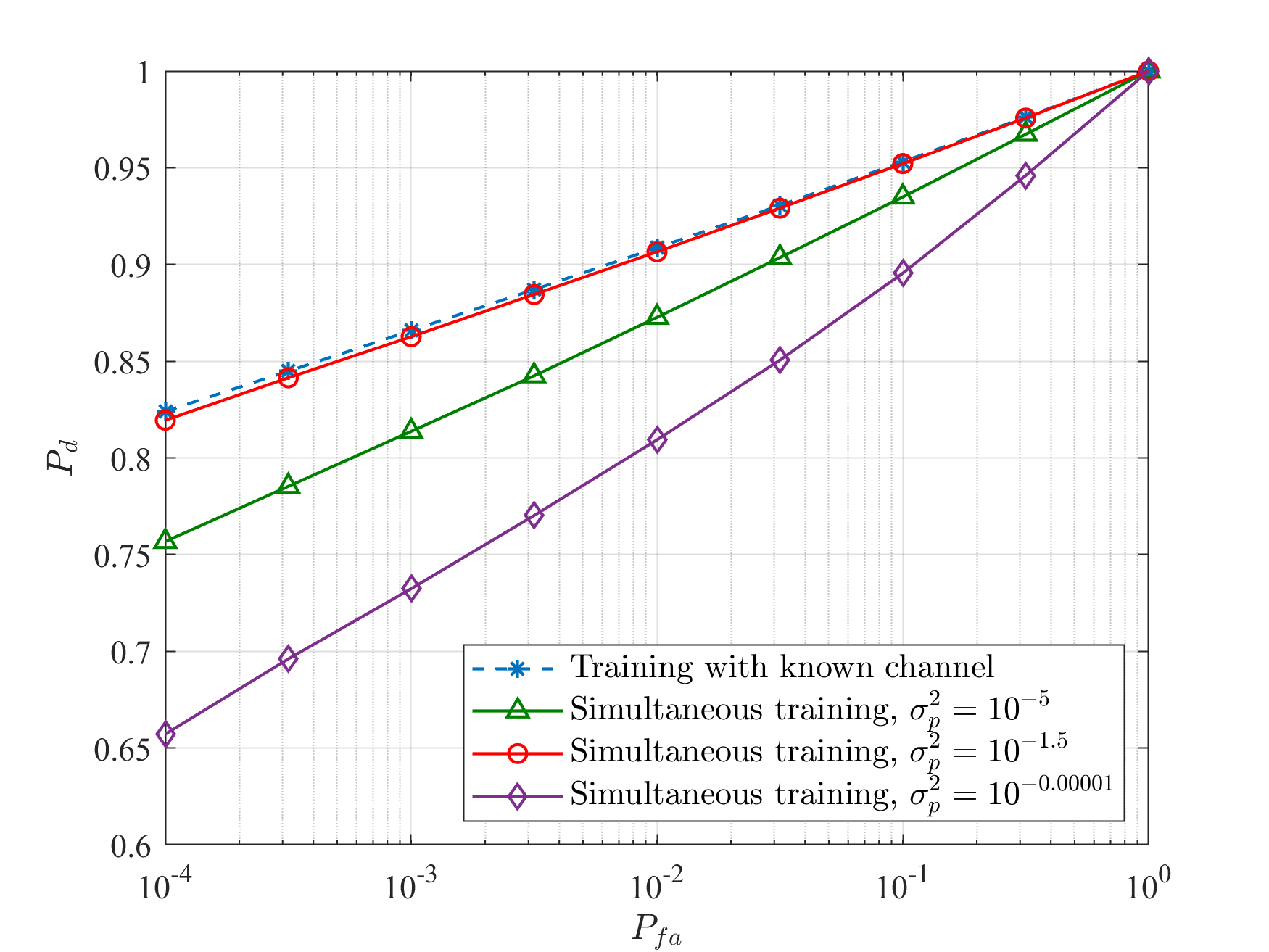}  
	\vspace{-3ex} 
	\caption{ROC curves for training with known channel and simultaneous training with different values of policy parameter $\sigma^2_p$.} \label{f: var_ROC}
\end{figure}


\subsubsection{Simultaneous Training vs Alternating Training}
We now compare simultaneous and alternating training in terms of ROC curves in Fig. \ref{f: alter}. ROC curves based on the optimal detector in the NP sense, namely, the square law detector \cite{Richards 2005} and the adaptive/initialization waveform are plotted as benchmarks. As shown in the figure, simultaneous training provides a similar detection performance as alternating training. Furthermore, both simultaneous training and alternating training are seen to result in significant improvements as compared to training of only the receiver, and provide detection performance comparable to adaptive waveform \cite{Wei 2019NN} and square law detector.

\begin{figure}[H]
	\centering
	\vspace{-3ex} 
	\includegraphics[width=0.7\linewidth]{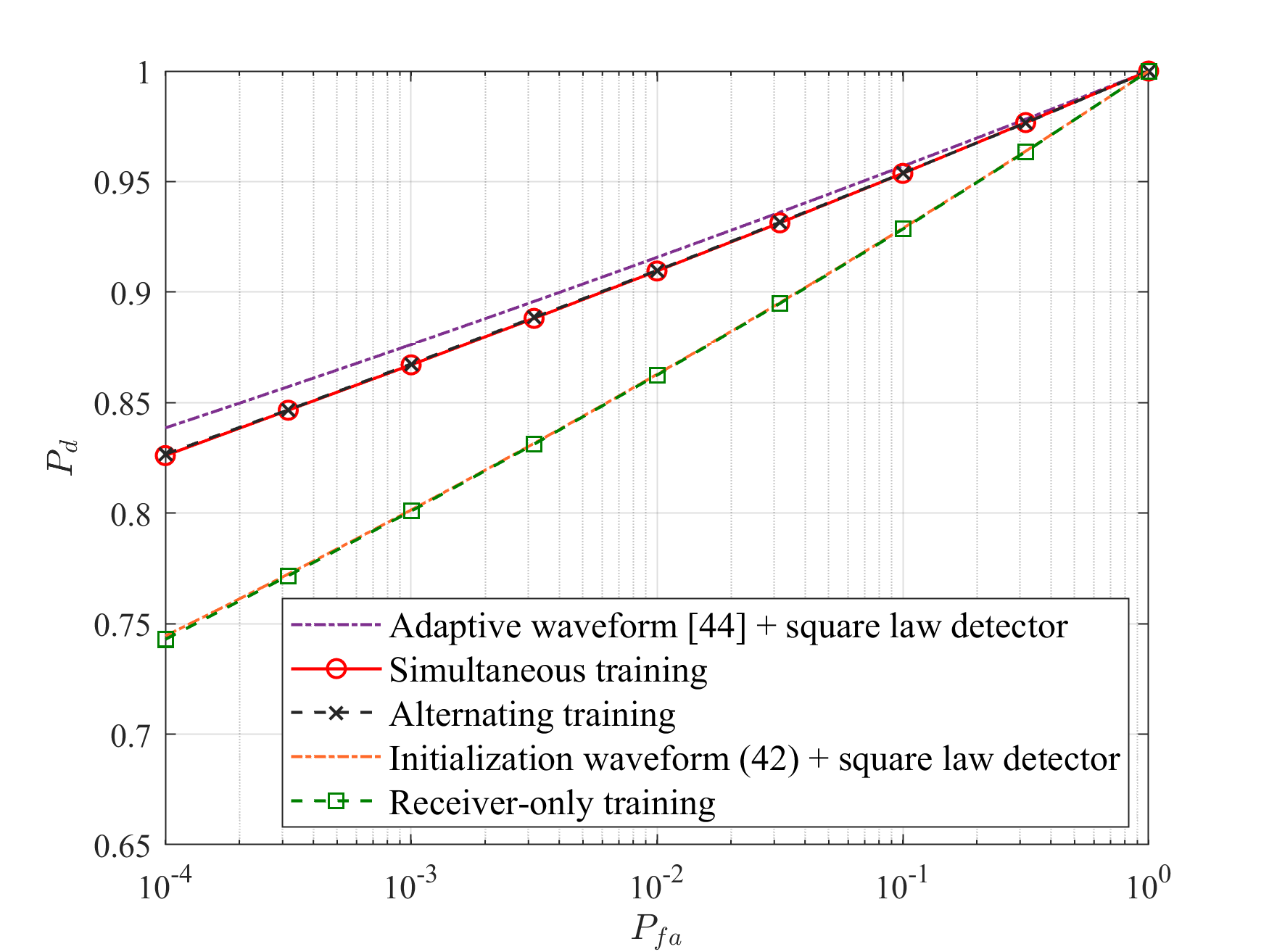}  
	\vspace{-3ex} 
	\caption{ROC curves with and without transmitter training.}
	\label{f: alter}
\end{figure}

\subsubsection{Learning Gaussian and Non-Gaussian Clutter} 
Two sets of ROC curves under different clutter statistics are illustrated in Fig. \ref{f: non-G1}. Each set contains two ROC curves with the same clutter statistics: one curve is obtained based on simultaneous training, and the other one is based on model-based design. For simultaneous training, the shape parameter of the clutter distribution (\ref{eq: Weibull pdf}) in the training stage is the same as that in the test stage, i.e, $\beta_{\text{train}}=\beta_{\text{test}}$. In the test stage, for Gaussian clutter ($\beta_{\text{test}}=2$), the model-based ROC curve is obtained by the adaptive waveform and the optimal detector in the NP sense. As expected, simultaneous training provides a comparable detection performance with the adaptive waveform and square law detector (also shown in Fig. \ref{f: alter}). In contrast, when the clutter is non-Gaussian ($\beta_{\text{test}}=0.25$), the optimal detector in the NP sense is mathematically intractable. Under this scenario, the data-driven approach is beneficial since it relies on data rather than a model. As observed in the figure, for non-Gaussian clutter with a shape parameter $\beta_{\text{test}}=0.25$, simultaneous training outperforms the adaptive waveform and square law detector.

\begin{figure}[H]
	\centering
	\vspace{-3ex} 
	\includegraphics[width=0.7\linewidth]{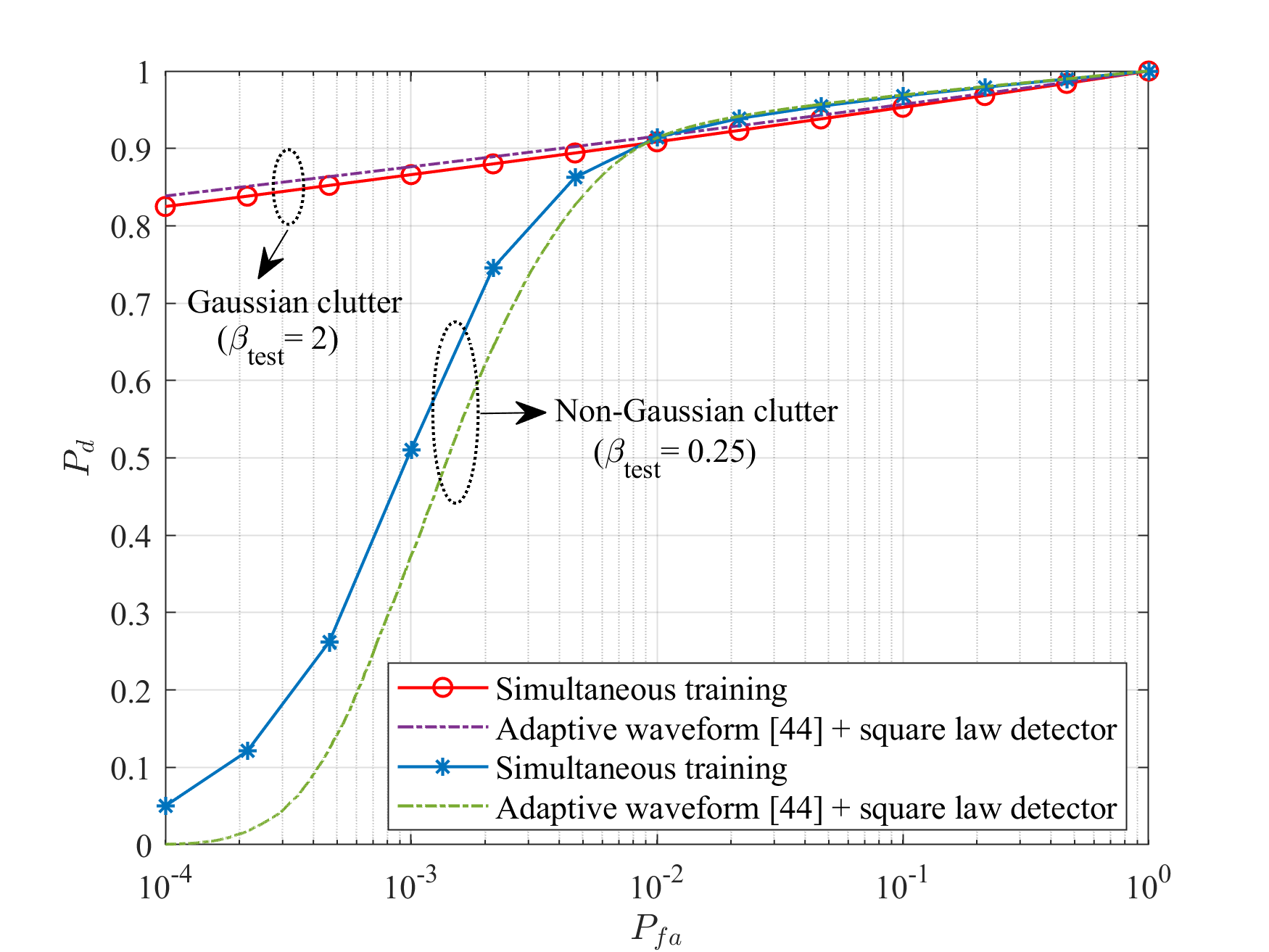}  
	\vspace{-3ex} 
	\caption{ROC curves for Gaussian/non-Gaussian clutter. The end-to-end radar system is trained and tested by the same clutter statistics, i.e, $\beta_{\text{train}}=\beta_{\text{test}}$.}
	\label{f: non-G1}
\end{figure}

\subsubsection{Simultaneous Training with Mixed Clutter Statistics} 
The robustness of the trained radar system to the clutter statistics is investigated next. As discussed previously, model-based design relies on a single clutter model, whereas data-driven learning depends on a training dataset. The dataset may contain samples from multiple clutter models. Thus, the 
system based on data-driven learning may be robustified by drawing samples from a mixture of clutter models. In the test stage, the clutter model may not be the same as any of the clutter models used in the training stage. As shown in the figure, for simultaneous training, the training dataset contains clutter samples generated from (\ref{eq: Weibull pdf}) with four different values of shape parameter $\beta_{\text{train}}\in \{0.25, 0.5, 0.75, 1\}$. The test data is generated with a clutter shape parameter $\beta_{\text{test}}=0.3$ not included in the training dataset. The end-to-end leaning radar system trained by mixing clutter samples provides performance gains compared to a model-based system using an adaptive waveform and square law detector.

\begin{figure}[H]
	\centering
	\vspace{-3ex} 
	\includegraphics[width=0.7\linewidth]{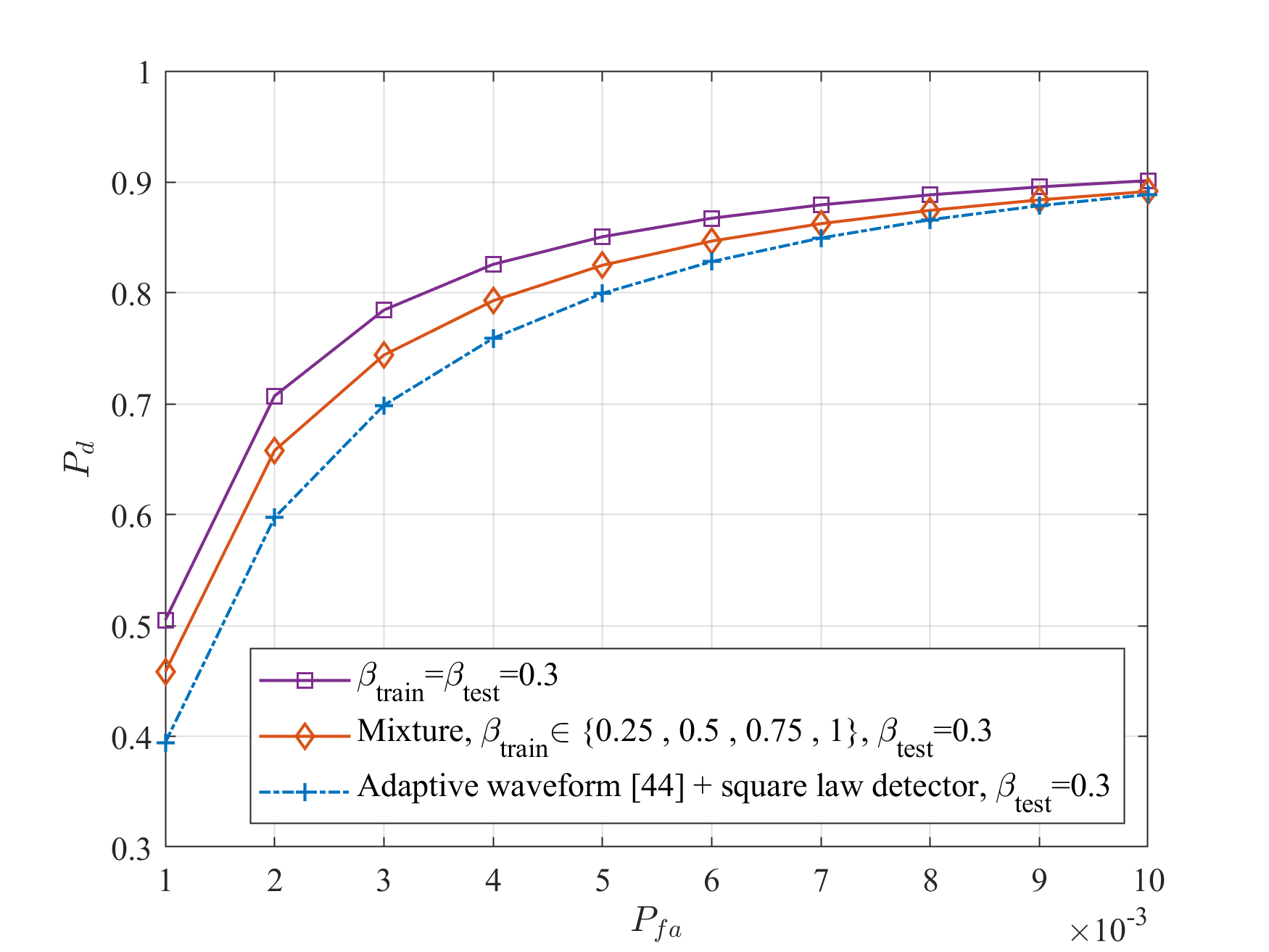}  
	\vspace{-3ex} 
	\caption{ROC curves for non-Gaussian clutter. To robustify detection performance, the end-to-end leaning radar system is trained with mixed clutter statistics, while testing for a clutter model different than used for training.}
	\label{f: non-mix}
\end{figure}

\subsubsection{Simultaneous Training under PAR Constraint} Detection performance with waveforms learned subject to a PAR constraint is shown in Fig. \ref{f:PAPR_fig1}. The end-to-end system trained with no PAR constraint, i.e., $\lambda=0$, serves as the reference. It is observed the detection performance degrades as the value of the penalty parameter $\lambda$ increases. Moreover,
PAR values of waveforms with different $\lambda$ are shown in Table \ref{table:3}. As shown in Fig. \ref{f:PAPR_fig1} and Table \ref{table:3}, there is a tradeoff between detection performance and PAR level. For instance,
given $P_{fa}=5\times 10^{-4}$, training the transmitter with the largest penalty parameter $\lambda=0.1$ yields the lowest $P_d=0.852$ with the lowest PAR value $0.17$ dB. In contrast, training the transmitter with no PAR constraint, i.e., $\lambda=0$, yields the best detection with the largest PAR value $3.92$ dB. 
Fig. \ref{f:PAPR_fig2} compares the normalized modulus of waveforms with different values of the penalty parameter $\lambda$. As shown in Fig. \ref{f:PAPR_fig2} and Table \ref{table:3}, the larger the penalty parameter $\lambda$ adopted in the simultaneous training, the smaller the PAR value of the waveform. 
\begin{table}[H]
	\caption{PAR values of waveforms with different values of
		penalty parameter $\lambda$}
	\label{table:3}
	\centering
	\resizebox{0.5\columnwidth}{!}{
		\begin{tabular}{@{}cccc@{}}\toprule
			& $\lambda=0$ (reference) & $\lambda=0.01$ & $\lambda=0.1$ \\ 
			\cmidrule{2-4}
			PAR [dB] (\ref{eq: PAPR complex}) & 3.92 & 1.76 & 0.17 \\
			\bottomrule
		\end{tabular}
	}
\end{table} 

\begin{figure}[H]
	\centering
	\vspace{-3ex} 
	\includegraphics[width=0.68\linewidth]{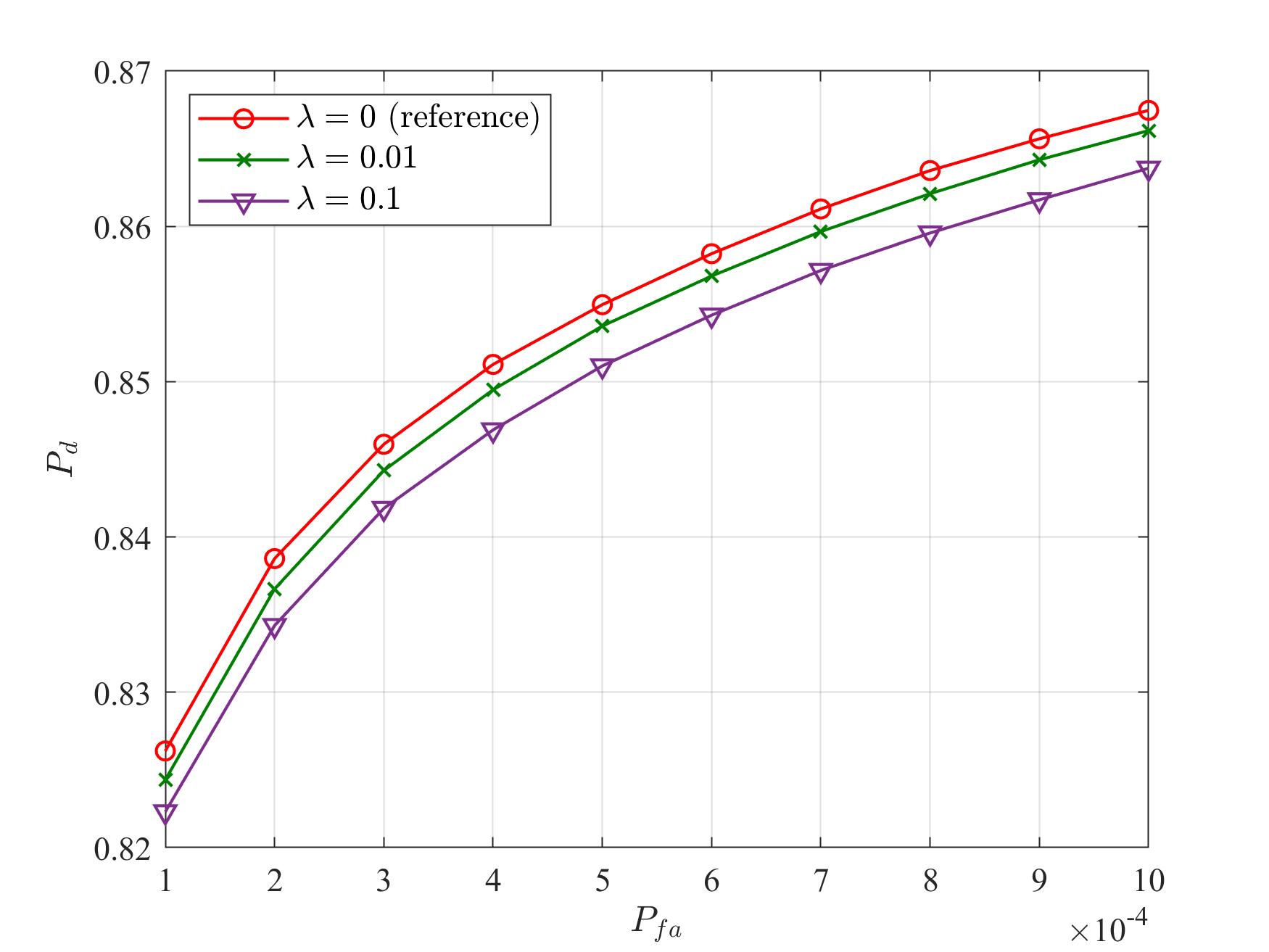}  
	\vspace{-3ex} 
	\caption{ROC curves for PAR constraint with the different
		values of the penalty parameter $\lambda$.}
	\label{f:PAPR_fig1}
\end{figure}

\begin{figure}[H]
	\centering
	\vspace{-3ex} 
	\includegraphics[width=0.68\linewidth]{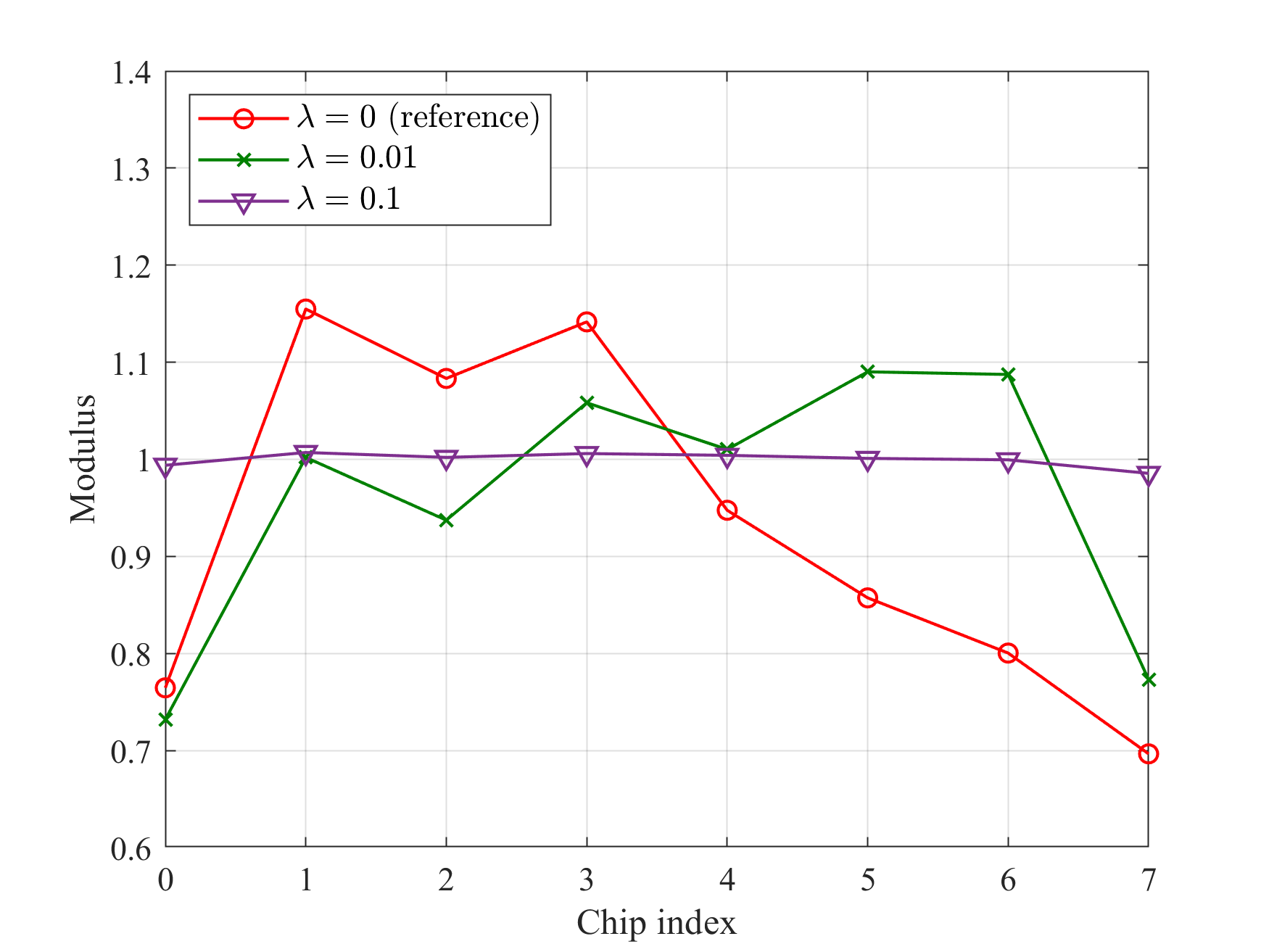}  
	\vspace{-3ex} 
	\caption{Normalized modulus of transmitted waveforms with different values of penalty parameter $\lambda$.}
	\label{f:PAPR_fig2}
\end{figure}

\subsubsection{Simultaneous Training under Spectral Compatibility Constraint} 
ROC curves for spectral compatibility constraint with different values of the penalty parameter $\lambda$ are illustrated in Fig. \ref{f:spectrum_fig1}.
The shared frequency bands are $\Gamma_1=[0.3,0.35]$ and $\Gamma_2=[0.5,0.6]$.
The end-to-end system trained with no spectral compatibility constraint, i.e., $\lambda=0$, serves as the reference. Training the transmitter with a large value of the penalty parameter $\lambda$ is seen to result in performance degradation. Interfering energy from radar waveforms trained with different values of $\lambda$ are shown in Table \ref{table:4}. It is observed that $\lambda$ plays an important role in controlling the tradeoff between detection performance and spectral compatibility of the waveform. For instance, for a fixed $P_{fa}=5 \times 10^{-4}$, training the transmitter with $\lambda=0$ yields $P_d=0.855$ with an amount of interfering energy $-5.79$ dB on the shared frequency bands, while training the transmitter with $\lambda=1$ creates notches in the spectrum of the transmitted waveform at the shared frequency bands.
Energy spectral densities of transmitted
waveforms with different values of $\lambda$ are illustrated in Fig. \ref%
{f:spectrum_fig2}. A larger the penalty parameter $\lambda$ results in a lower amount of interfering energy in the prescribed frequency shared
regions. Note, for instance, that the nulls of the energy spectrum density of the waveform for $%
\lambda=1$ are much deeper than their counterparts for $\lambda=0.2$. 

\begin{figure}[H]
	\centering
	\vspace{-3ex}
	\includegraphics[width=0.7\linewidth]{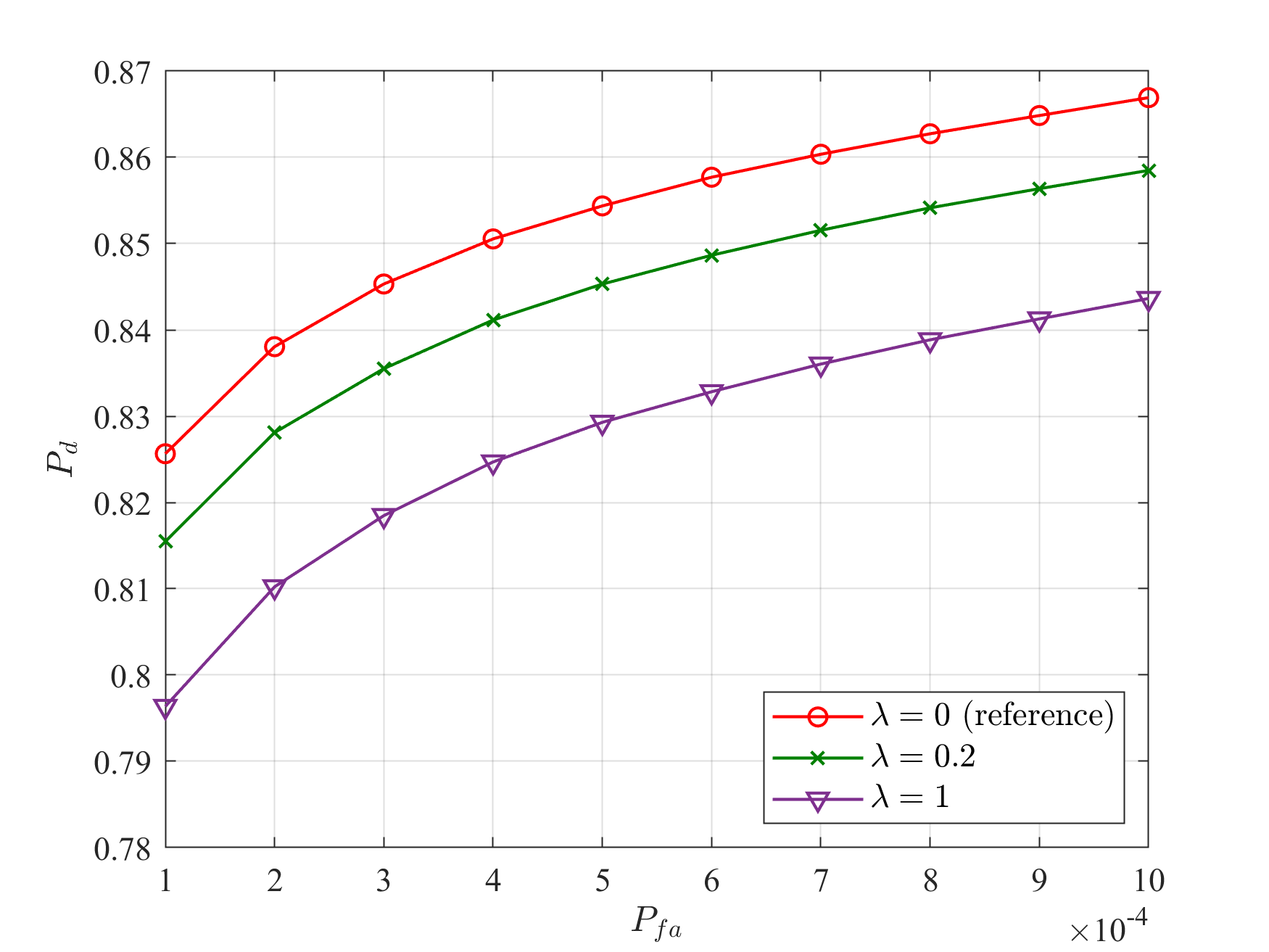}  
	\vspace{-3ex}
	\caption{ROC curves for spectral compatibility constraint for different values of penalty parameter $\lambda$.}
	\label{f:spectrum_fig1}
\end{figure}

\begin{table}[H]
	\caption{Interfering energy from radar waveforms with different values of
		weight parameter $\lambda$ }
	\label{table:4}
	\centering
	\resizebox{0.65\columnwidth}{!}{
		\begin{tabular}{@{}cccc@{}}\toprule
			& $\lambda=0$ (reference) & $\lambda=0.2$ & $\lambda=1$ \\ 
			\cmidrule{2-4}
			Interfering energy [dB] (\ref{eq: spectrum complex})  & -5.79 & -10.39& -17.11 \\
			\bottomrule
		\end{tabular}
	}
\end{table}
\begin{figure}[H]
	\centering
	\vspace{-3ex}
	\includegraphics[width=0.7\linewidth]{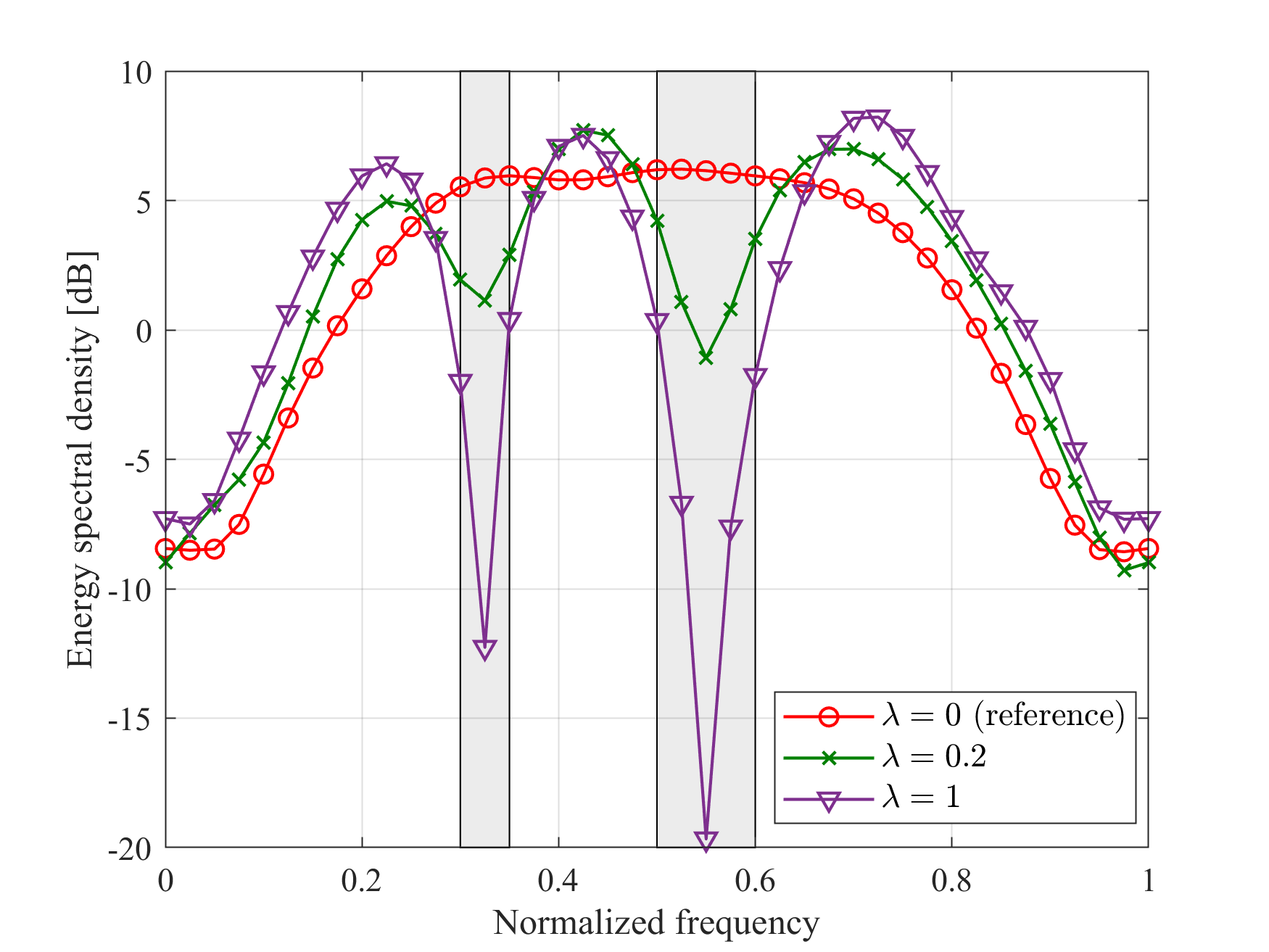}  
	\vspace{-3ex}
	\caption{Energy spectral density of waveforms with different values of penalty parameter $\lambda$.}
	\label{f:spectrum_fig2}
\end{figure}

\section{Conclusions}

In this paper, we have formulated the radar design problem as end-to-end
learning of waveform generation and detection. We have developed two
training algorithms, both of which are able to incorporate various waveform
constraints into the system design. Training may be implemented either as simultaneous supervised training of the receiver and RL-based training of the transmitter, or as alternating between training of the receiver and of the transmitter. Both training algorithms have similar performance. We have also robustified the detection performance by training the system with
mixed clutter statistics. Numerical results have shown that the proposed
end-to-end learning approaches are beneficial under non-Gaussian clutter, and
successfully adapt the transmitted waveform to actual statistics of
environmental conditions, while satisfying operational constraints.

\numberwithin{equation}{section} \appendices
\section{Gradient of Penalty Functions}
In this appendix are derived the respective gradients of the penalty functions (\ref{eq: PAPR complex}) and (\ref{eq: spectrum complex}) with respect to the transmitter parameter vector $\boldsymbol{\theta}_T$. To facilitate the presentation, let $\overline{\mathbf{y}}_{\boldsymbol{\theta}_T}$ represent a $2K \times 1$ real vector comprising the real and imaginary parts of the waveform $\mathbf{y}_{\boldsymbol{\theta}_T}$, i.e., $\overline{\mathbf{y}}_{\boldsymbol{\theta}_T}=\big[\Re(\mathbf{y}_{\boldsymbol{\theta}_T}), \Im (\mathbf{y}_{\boldsymbol{\theta}_T})\big]^T$. 
\subsubsection{Gradient of PAR Penalty Function}	
	As discussed in Section II-B, the transmitted power is normalized such that $||\mathbf{y}_{\boldsymbol{\theta}_T} ||^2=||\overline{\mathbf{y}}_{\boldsymbol{\theta}_T}||^2=1$. Let subscript ``max'' represent the chip index associated with the PAR value (\ref{eq: PAPR complex}). By leveraging the chain rule, the gradient of (\ref{eq: PAPR complex}) with respect to $\boldsymbol{\theta}_T$ is written
	\begin{equation}
		\nabla_{\boldsymbol{\theta}_T}J_{\text{PAR}}(\boldsymbol{\theta}_T)=\nabla_{\boldsymbol{\theta}_T}\overline{\mathbf{y}}_{\boldsymbol{\theta}_T} \cdot \mathbf{g}_{\text{PAR}},
	\end{equation}
	where $\mathbf{g}_{\text{PAR}}$ represents the gradient of the PAR penalty function $J_{\text{PAR}}(\boldsymbol{\theta}_T)$ with respect to $\overline{\mathbf{y}}_{\boldsymbol{\theta}_T}$, and is given by
	\begin{equation}
		\mathbf{g}_{\text{PAR}}=\big[
		\begin{array}{c;{2pt/2pt}c}
			0,\dots ,0, 2K\Re({{y}}_{\boldsymbol{\theta}_T,\text{max}}),0, \dots, 0 & 0,\dots, 0, 2K\Im({{y}}_{\boldsymbol{\theta}_T,\text{max}}),0, \dots, 0
		\end{array}
		\big]^T.
	\end{equation}  
\subsubsection{Gradient of Spectral Compatibility Penalty Function}	According to the chain rule, the gradient of (\ref{eq: spectrum complex}) with respect to $\boldsymbol{\theta}_T$ is expressed
	\begin{equation}
		\nabla_{\boldsymbol{\theta}_T}J_{\text{spectrum}}(\boldsymbol{\theta}_T)=\nabla_{\boldsymbol{\theta}_T}\overline{\mathbf{y}}_{\boldsymbol{\theta}_T} \cdot \mathbf{g}_{\text{spectrum}},
	\end{equation}
	where $\mathbf{g}_{\text{spectrum}}$ denotes the gradient of the spectral compatibility penalty function $J_{\text{spectrum}}(\boldsymbol{\theta}_T)$ with respect to $\overline{\mathbf{y}}_{\boldsymbol{\theta}_T}$, and is given by
	\begin{equation}
		\mathbf{g}_{\text{spectrum}}=\left[
		\begin{array}{c}
			2\Re\big[(\boldsymbol{\Omega}\mathbf{y}_{\boldsymbol{\theta}_T})^*\big] \\ \hdashline[2pt/2pt] -2\Im\big[(\boldsymbol{\Omega}\mathbf{y}_{\boldsymbol{\theta}_T})^*\big]
		\end{array}
		\right].
\end{equation} 

\section{Proof of Proposition 1}
\begin{proof}
	The average loss function of simultaneous
	 training $\mathcal{L}^{\pi}(\boldsymbol{\theta}_R, \boldsymbol{\theta}_T)$ (\ref{eq: joint loss}) could be expressed
	\begin{equation}
		\mathcal{L}^{\pi}(\boldsymbol{\theta}_R, \boldsymbol{\theta}_T)=\sum_{i\in\{0,1\}}P(\mathcal{H}_i) \int_{\mathcal{A}}\pi(\mathbf{a}|\mathbf{y}_{\boldsymbol{\theta}_T})\int_{\mathcal{Z}} \ell \big( f_{\boldsymbol{\theta}_R}(\mathbf{z}),i\big)p(\mathbf{z}|\mathbf{a},\mathcal{H}_i)d\mathbf{z}d\mathbf{a}.  \label{a: fuse loss ori.}
	\end{equation}
	As discussed in Section II-B, the last layer of the receiver implementation consists of a sigmoid activation function, which leads to the output of the receiver $f_{\boldsymbol{\theta}_R}(\mathbf{z})\in (0,1)$. Thus there exist a constant $b$ such that $\sup_{\mathbf{z},i} \ell \big( f_{\boldsymbol{\theta}_R}(\mathbf{z}),i\big) <b<\infty$. 
	Furthermore, for $i\in \{0,1\}$, the instantaneous values of the cross-entropy loss $\ell \big( f_{\boldsymbol{\theta}_R}(\mathbf{z}),i\big)$, the policy $\pi(\mathbf{a}|\mathbf{y}_{\boldsymbol{\theta}_T})$, and the likelihood $p(\mathbf{z}|\mathbf{a},\mathcal{H}_i)$ are continuous in variables $\mathbf{a}$ and $\mathbf{z}$. By leveraging Fubini's theorem \cite{Fubini} to exchange the order of integration in (\ref{a: fuse loss ori.}), we have
	\begin{equation}
		\begin{aligned}
			\mathcal{L}^{\pi}(\boldsymbol{\theta}_R, \boldsymbol{\theta}_T)=\sum_{i\in\{0,1\}}P(\mathcal{H}_i)\int_{\mathcal{Z}} \ell \big( f_{\boldsymbol{\theta}_R}(\mathbf{z}),i\big) \int_{\mathcal{A}}p(\mathbf{z}|\mathbf{a},\mathcal{H}_i) \pi(\mathbf{a}|\mathbf{y}_{\boldsymbol{\theta}_T}) d\mathbf{a}d\mathbf{z}.
		\end{aligned} \label{a: fuse loss exchage}
	\end{equation}
	Note that for a waveform $\mathbf{y}_{\boldsymbol{\theta}_T}$ and a target state indicator $i$, the product between the likelihood $p(\mathbf{z}|\mathbf{a},\mathcal{H}_i)$ and the policy $\pi(\mathbf{a}|\mathbf{y}_{\boldsymbol{\theta}_T})$ becomes a joint PDF of two 
	random variables $\mathbf{a}$ and $\mathbf{z}$, namely,
	\begin{equation}
		p(\mathbf{z}|\mathbf{a},\mathcal{H}_i)\pi(\mathbf{a}|\mathbf{y}_{\boldsymbol{\theta}_T})=p(\mathbf{a},\mathbf{z}|\mathbf{y}_{\boldsymbol{\theta}_T},\mathcal{H}_i). \label{a: joint prob.}
	\end{equation}
	Substituting (\ref{a: joint prob.}) into (\ref{a: fuse loss exchage}), we obtain 
	\begin{equation}
		\begin{aligned}
			\mathcal{L}^{\pi}(\boldsymbol{\theta}_R, \boldsymbol{\theta}_T)&=\sum_{i\in\{0,1\}}P(\mathcal{H}_i)\int_{\mathcal{Z}}\ell\big( f_{\boldsymbol{\theta}_R}(\mathbf{z}),i\big)\int_{\mathcal{A}}p(\mathbf{a},\mathbf{z}|\mathbf{y}_{\boldsymbol{\theta}_T},\mathcal{H}_i) d\mathbf{a}d\mathbf{z}\\
			&=\sum_{i\in\{0,1\}}P(\mathcal{H}_i)\int_{\mathcal{Z}}\ell\big( f_{\boldsymbol{\theta}_R}(\mathbf{z}),i\big)p(\mathbf{z}|\mathbf{y}_{\boldsymbol{\theta}_T},\mathcal{H}_i)d\mathbf{z}, 
		\end{aligned} \label{a: fuse loss final}
	\end{equation}
	where the second equality holds by integrating the joint PDF $p(\mathbf{z},\mathbf{a}|\mathbf{y}_{\boldsymbol{\theta}_T},\mathcal{H}_i)$ over the random variable $\mathbf{a}$, i.e.,
	$\int_{\mathcal{A}}p(\mathbf{a},\mathbf{z}|\mathbf{y}_{\boldsymbol{\theta}_T},\mathcal{H}_i) d\mathbf{a}=p(\mathbf{z}|\mathbf{y}_{\boldsymbol{\theta}_T},\mathcal{H}_i)$.
	
	Taking the gradient of (\ref{a: fuse loss final}) with respect to $\boldsymbol{\theta}_R$, we have
	\begin{equation}
		\begin{aligned}
			\nabla_{\boldsymbol{\theta }_R} \mathcal{L}^{\pi}(\boldsymbol{\theta}_R, \boldsymbol{\theta}_T)&= \sum_{i\in\{0,1\}}P(\mathcal{H}_i)\int_{\mathcal{Z}}p(\mathbf{z}|\mathbf{y}_{\boldsymbol{\theta}_T},\mathcal{H}_i)\nabla_{\boldsymbol{\theta }_R}\ell\big( f_{\boldsymbol{\theta}_R}(\mathbf{z}),i\big)d\mathbf{z}\\
			&={\nabla}_{\boldsymbol{\theta}_R}\mathcal{L}_R(\boldsymbol{\theta}_R),
		\end{aligned}
	\end{equation}
	where the second equality holds via (\ref{eq: rx loss grad.}). Thus, the proof of Proposition 1 is completed.
\end{proof}

\section{Proof of Proposition 2}
\begin{proof}
	According to (\ref{a: fuse loss final}), the gradient of the average loss function of simultaneous
	 training with respect to $\boldsymbol{\theta}_T$ is given by
	\begin{equation}
		\begin{aligned}
			\nabla_{\boldsymbol{\theta }_T} \mathcal{L}^{\pi}(\boldsymbol{\theta}_R, \boldsymbol{\theta}_T)&=\sum_{i\in\{0,1\}}P(\mathcal{H}_i)\int_{\mathcal{Z}}\ell\big( f_{\boldsymbol{\theta}_R}(\mathbf{z}),i\big) \nabla_{\boldsymbol{\theta }_T} p(\mathbf{z}|\mathbf{y}_{\boldsymbol{\theta}_T},\mathcal{H}_i)d\mathbf{z}\\
			& =  \nabla_{\boldsymbol{\theta }_T} \mathcal{L}(\boldsymbol{\theta}_R, \boldsymbol{\theta}_T),
		\end{aligned} \label{a: fuse loss tx grad. ori.}
	\end{equation}
where the last equality holds by (\ref{eq: tx loss known grad}). The proof of {Proposition 2} is completed. 
\end{proof}

\end{document}